\documentclass[10pt,journal]{IEEEtran}
\usepackage{amsthm,amsmath}
\usepackage{booktabs}
\usepackage{pgfplots}
\usepackage{epstopdf}
\usepackage{graphicx}
\usepackage{amsfonts}
\usepackage{amssymb}

\usepackage{cool}
\usepackage{float}
\usepackage{tikz}
\usepackage{url}
\usetikzlibrary{shapes,arrows,backgrounds}
\usetikzlibrary{calc,fit,arrows}
\usepackage{caption}
\usepackage{subcaption}
\usepackage{adjustbox}
\usepackage{lineno}
\usepackage{setspace}
\usepackage{stfloats} 
\usepackage{array, multirow}
\newcolumntype{C}[1]{>{\centering\arraybackslash}p{#1}}
\usetikzlibrary{shapes}
\usepgfplotslibrary{groupplots}
\pgfplotsset{compat=1.13}
\usepackage{setspace}
\usetikzlibrary{intersections}
\usepgfplotslibrary{fillbetween}
\newtheorem{theorem}{Theorem}

\newtheorem{lemma}{Lemma}

\newtheorem{prop}{Proposition}

\newtheorem{remark}{Remark}

\DeclareMathOperator{\diag}{diag}
\DeclareMathOperator{\tr}{Tr}

\allowdisplaybreaks

\begin{document}


\title{Capacity Bounds for Optical Interference Channels Subject to Average Power Constraints}%

\author{Viswanathan~Ramachandran
\thanks{
The author is with the Department of Electronics and Electrical Communication Engineering, Indian Institute of Technology Kharagpur, West Bengal 721302, India (email: \protect\url{vramachandran@ece.iitkgp.ac.in}).}
}
\maketitle

\begin{abstract}
Accurate information-theoretic analysis of multiuser wavelength-division multiplexed (WDM) optical links is complicated by Kerr nonlinearity and chromatic dispersion. Existing capacity bounds for first-order regular perturbative (RP) channel models have relied on per-symbol peak power constraints to simplify the underlying extremal analysis. In this paper, we replace these peak constraints with the more practically relevant average input power constraints and study the resulting RP WDM interference channel. Unlike the peak-constrained formulation, an exact finite signal-to-noise ratio (SNR) analysis under average power leads in general to a nonconvex, high-dimensional optimization problem with no closed-form solution.

We show that the problem admits a tractable characterization in the high-SNR regime. Using a genie-aided argument, we derive closed-form single-user high-SNR capacity bounds for the RP interference channel and identify constant-amplitude interferer signaling as extremal. Furthermore, under an empirically validated non‑negativity condition on the dominant RP kernels and within a truncated first‑order RP model, we establish a finite‑SNR extremality result, yielding a single‑user capacity characterization within the RP model under average power constraints. Building on these single-user corner points, we construct an \emph{all-helps-one} achievable region via time sharing, which strictly improves upon treating nonlinear interference as noise. Numerical results and split-step Fourier simulations of the nonlinear Schrödinger equation channel confirm that the qualitative conclusions persist over practical launch powers. The results provide new insight into extremal signaling and multiuser trade-offs for nonlinear WDM systems within the RP approximation.
\end{abstract}

\begin{IEEEkeywords}
Wavelength division multiplexing, interference channels, capacity region, extremal analysis, high-SNR.
\end{IEEEkeywords}

\section{Introduction}
Nonlinear interference caused by the optical \emph{Kerr effect} fundamentally couples users together in wavelength-division multiplexed (WDM) links and limits the achievable data rates on long-haul fibre systems, especially at high values of launch powers. The optical field corresponding to each wavelength channel both experiences as well as induces cross-phase modulation (XPM), so that the interference seen by any user depends upon the launch signals of all users. In modern dispersion-uncompensated links, XPM is the dominant impairment and is usually treated at the receiver as additional noise after matched filtering~\cite{secondini2013achievable,dar2013properties}. While practical digital backpropagation and dispersion compensation can reduce intra-channel effects such as self-phase modulation (SPM), they cannot compensate for inter-channel nonlinearities such as XPM.

Information-theoretic methods have been increasingly applied to quantify fundamental limits of optical transmission~\cite{essiambre2010capacity,kramer2015upper,agrell2015influence,secondini2016scope}. Most studies, however, analyze the achievable rates for a single WDM user while assuming fixed, often Gaussian, statistics for the interfering users. For example, \cite{secondini2016scope} showed that the single-channel capacity under ideal dispersion management grows unboundedly with launch power, in contrast to Gaussian achievable information rates that saturate. Reference \cite{agrell2015influence} examined how different behavioural models for the interfering users alter the capacity of a given user of interest. Outer bounds on the single-user nonlinear Schr\"odinger equation (NLSE) channel have been given in \cite{kramer2015upper}, but no multiuser region characterizations followed.

Within multiuser information theory, channels where several transmitter–receiver pairs interfere but do not cooperate are modelled as interference channels~\cite[Ch.~6]{el2011network}. This abstraction is particularly relevant to WDM systems, where wavelengths typically carry independent data streams and joint encoding or coordinated transmission across transceivers is generally infeasible. Optical interference channels have received limited attention in the literature. While \cite{taghavi2006multiuser} recast WDM as a multiple-access channel with full receiver cooperation, and \cite{ghozlan2017models} introduced interference focusing under a simplified dispersionless model, the simultaneously achievable rate region for realistic WDM channels remains largely uncharacterized.

A recent attempt to bridge this gap is the author’s prior work~\cite{ramachandran2023capacity,viswanathan2021wdm,viswanathan2022wdm}, which considered the first-order regular perturbative (RP) model and derived inner and outer bounds by imposing per-symbol \emph{peak} power constraints. Peak constraints simplify the extremal signalling analysis but are not the operational metric in optical systems, where average launch power per channel is specified. As a result, the capacity region under the practically relevant \emph{average} power constraints has remained open, with no inner/outer bounds available.

This paper addresses this crucial gap. Unlike per-symbol peak constraints, the average-power formulation leads at finite signal-to-noise ratio (SNR) to a nonconvex, high-dimensional optimization problem for which no closed-form capacity characterization is possible in general. The main contribution of this paper is to show that, despite the inherent non-convexity of the problem, the extremal signaling structure can be identified in the high-SNR regime and, under an empirically validated sign condition, also at finite SNR. Although the channel model coincides with \cite{ramachandran2023capacity}, none of the extremal or capacity characterizations derived here can be obtained by a post-hoc relaxation of peak constraints, and requires different arguments based on convex-analysis and extremal geometry. Specifically, we derive, for the first time, closed-form high-SNR capacity region outer bounds with coordinate-wise matching inner bounds for the RP WDM interference channel under average power constraints. Moreover, under a non-negativity condition on the perturbative kernels, we establish a non-asymptotic extremality result yielding a finite-SNR single-user capacity corner point characterization.\footnote{All capacity characterizations in this paper are understood to be with respect to the first‑order RP model \eqref{eq:Kuserapprox}; they do not claim to describe the physical NLSE channel beyond the validity regime of the RP approximation.} 

\noindent\textbf{Main Contributions:}
\begin{itemize}
\item We derive a new genie-aided outer bound for the WDM interference channel under average power constraints in the high-SNR regime, in closed form (Theorem~\ref{thm:OBgenieK}).

\item We characterize the single-user high-SNR capacity corner points and identify constant-amplitude interferer signaling to be optimal/extremal (Proposition~\ref{prop:extreme}).

\item Under an empirically validated non-negativity condition on the dominant RP kernels, we establish a finite-SNR extremality result, yielding a finite‑SNR single‑user capacity corner point characterization (Propositions~\ref{thm:OBgenieKfinite}
and~\ref{prop:finiteSNR-extreme}).

\item We construct an achievable region via time sharing among the single-user capacity corner points, termed the \emph{all-helps-one} (AHO) scheme, which strictly improves upon treating interference as noise (Proposition~\ref{prop:allhelpsone}).
\end{itemize}

The significance of the asymptotic characterizations is that the high-SNR regime is precisely where conventional TIN approaches break down. Moreover, the finite-SNR characterizations closely match achievable rates at practical launch powers 0--5\,dBm, providing both engineering insight and new theoretical limits. Together, these results extend earlier peak-constrained analyses and provide a tractable framework for understanding nonlinear interference in WDM systems.

\emph{Paper organisation:} Section~\ref{sec:model} introduces the channel model. Outer bounds and single-user high-SNR corner points are derived in Section~\ref{sec:Kuser}, and a finite-SNR characterization is established in Section~\ref{sec:Kuserfinite}. Section~\ref{sec:ach} presents the all-helps-one achievable region. Proofs of the main technical results are deferred to the appendices. Numerical results and conclusions appear in Sections~\ref{sec:num} and~\ref{sec:conc}, respectively.

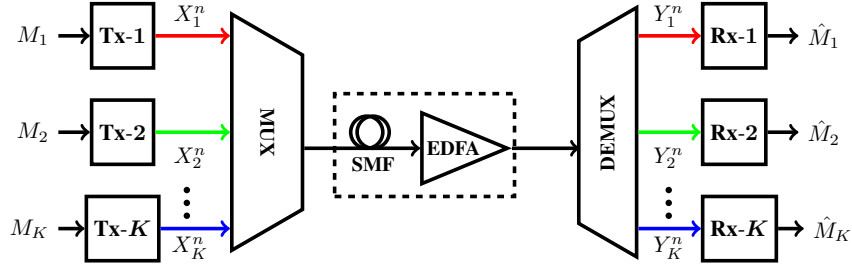
\begin{figure*}[t]
\begin{center}
\scalebox{0.85}{\begin{tikzpicture}[ultra thick]
\node (e1) at (-2.5,1.5) [rectangle, draw, minimum height=1.0cm]{\textbf{Tx}-$\boldsymbol{1}$}; 
\node (e2) at (-2.5,-1.5) [rectangle, draw, minimum height=1.0cm]{\textbf{Tx}-$\boldsymbol{K}$};
\node (e3) at (-2.5,0) [rectangle, draw, minimum height=1.0cm]{\textbf{Tx}-$\boldsymbol{2}$};
\foreach \i in {-1.3,-1.1,-0.9} {\fill (-1.5,\i) circle(0.05cm);}
\node[trapezium,draw,rotate=-90,minimum height=1.1cm] (t1) at (-0.25,0) {\small \textbf{MUX}};
\node[trapezium,draw,rotate=90,minimum height=0.9cm,minimum width=0.3cm] (t2) at (5.05,0) {\small \textbf{DEMUX}};
\node (c1) at (1.4,0) [circle, draw, minimum height=0.5cm]{};
\node (c2) at (1.3,0) [circle, draw, minimum height=0.5cm]{};
\node[isosceles triangle, draw, minimum size =1.1cm] (T)at (2.5,-0.25){};
\node (d1) at (6.5,1.5) [rectangle, draw, right, minimum height=1.0cm]{\textbf{Rx}-$\boldsymbol{1}$};
\node (d2) at (6.5,-1.5) [rectangle, draw, right, minimum height=1.0cm]{\textbf{Rx}-$\boldsymbol{K}$};
\node (d3) at (6.5,0) [rectangle, draw, right, minimum height=1.0cm]{\textbf{Rx}-$\boldsymbol{2}$};
\foreach \i in {-1.3,-1.1,-0.9} {\fill (6,\i) circle(0.05cm);}
\node () at (1.40,-0.5) {\small \textbf{SMF}};
\node () at (2.65,-0.25) {\small \textbf{EDFA}};
\node (e5) at (2.20,-0.2) [rectangle, dashed, draw, minimum height=1.6cm, minimum width=2.8cm]{};
\node () at (-1.45,1.8) {$X_1^n$};
\node () at (-1.45,-1.8) {$X_K^n$};
\node () at (-1.45,-0.4) {$X_2^n$};
\node () at (6,1.8) {$Y_1^n$};
\node () at (6,-1.8) {$Y_K^n$};
\node () at (6,-0.4) {$Y_2^n$};
\draw[<-] (e1) --++(-1.0,0) node[left]{$M_1$};
\draw[<-] (e2) --++(-1.0,0) node[left]{$M_K$};
\draw[<-] (e3) --++(-1.0,0) node[left]{$M_2$};
\draw[->] (d1) --++(1.0,0) node[right]{$\hat{M}_1$};
\draw[->] (d2) --++(1.0,0) node[right]{$\hat{M}_K$};
\draw[->] (d3) --++(1.0,0) node[right]{$\hat{M}_2$};
\draw[->] (T.east)  --++(1.05,0) node[above]{};
\draw[<-] (T.-180)  -- (T.180-|t1.90);
\draw[->,color=red] (e1) --++(1.67,0) node[above]{};
\draw[->,color=blue] (e2) --++(1.67,0) node[above]{};
\draw[->,color=green] (e3) --++(1.67,0) node[above]{};
\draw[<-,color=red] (d1.west) --++(-0.97,0) node[midway, above]{};
\draw[<-,color=blue] (d2.west) --++(-0.97,0) node[midway, above]{};
\draw[<-,color=green] (d3.west) --++(-0.97,0) node[midway, above]{};
\end{tikzpicture}}
\caption{Regular perturbative interference channel model for multiuser WDM transmission.}
\label{fig:ICmodel}
\end{center}
\end{figure*}

\section{Channel Model}\label{sec:model}
We adopt the standard first-order regular perturbative (RP) discrete-time model for nonlinear WDM links
\cite{secondini2013achievable,dar2013properties,mecozzi2012nonlinear}, obtained by expanding the noisy nonlinear Schr\"odinger equation (NLSE) to first order in the Kerr nonlinearity and sampling after ideal dispersion compensation and matched filtering. The model captures cross-phase modulation (XPM), which dominates nonlinear interference (NLI) in this regime, while assuming ideal compensation of self-phase modulation and neglecting higher-order signal-noise interactions.

The complex baseband field $A(t,z)$ along a single-polarization optical fiber obeys the NLSE~\cite{agrawal2000nonlinear}
\begin{align} \label{eq:nlse}
\frac{\partial A}{\partial z} = -\frac{\alpha}{2}A
+ \frac{\jmath}{2}\beta_2 \frac{\partial^2 A}{\partial \tau^2}
- \jmath \gamma |A|^2 A + W(t,z),
\end{align}
where $t$ and $z$ denote continuous time and propagation distance, respectively, $\tau=t-\beta_1 z$ is the retarded time frame, $\alpha$ denotes fiber attenuation, $\beta_1$ and $\beta_2$ are dispersion coefficients, $\gamma$ is the Kerr nonlinearity coefficient, and $W(\cdot)$ models amplified spontaneous emission noise.

Following the standard RP analysis \cite{mecozzi2012nonlinear,dar2013properties}, this leads to the discrete-time $K$-user WDM interference channel
\begin{align}\label{eq:Kuser-full}
Y_{k}[i] &= X_{k}[i]+N_{k}[i]+ \jmath \gamma \sum_{t,l,m} \sum_{\substack{w=1\\w\neq k}}^{K} S_{k,w}^{l,m,t} X_k[i-t] \notag\\
&\phantom{wwwwwwwwwwwwww} \times X_w[i-l]X_w[i-m]^* ,
\end{align}
where $X_k[i]$ and $Y_k[i]$ denote the discrete-time channel input and output of user (wavelength) $k$ at symbol time $i$, and $N_k[i]$ are i.i.d.\ circularly symmetric complex Gaussian noise samples with variance $2\sigma_k^2$. The indices $l,m,t$ represent discrete memory offsets due to dispersion and nonlinear interactions, and the coefficients $S_{k,w}^{l,m,t}$ are fixed RP kernels determined by the fiber parameters and span configuration~\cite{dar2013properties}. The Gaussian noise model is standard in coherent optical systems after dispersion compensation and matched filtering, arising from the accumulation of many independent spontaneous emission contributions~\cite{essiambre2010capacity}. At very high signal powers, nonlinear signal-noise interactions may lead to deviations from Gaussianity, which are not captured by the first-order RP model here.

Each user’s codewords must satisfy an \emph{average} input power constraint specified as
\begin{align}\label{eq:avgpower}
\frac{1}{n}\sum_{i=1}^{n}\mathbb{E}[|X_k[i]|^{2}] \le P_k, \qquad k=1,2,\dots,K.
\end{align}
This is the expected-power formulation of the standard blockwise power constraint $\frac{1}{n}\sum_{i=1}^n |x_k[i](m_k)|^2 \le P_k$ imposed on each codeword, where the expectation is taken with respect to the distribution of the message $M_k$ and the randomness in the codebook generation (the codeword being $X_k[i](M_k)$). It replaces the per-symbol peak constraint in~\cite{ramachandran2023capacity} and reflects the physically relevant operating condition in WDM systems.

Perturbative analyses show that the dominant nonlinear interference terms correspond to two-pulse collisions with $t=0$ and $l=m$ \cite{dar2016pulse,ramachandran2023capacity}. Accordingly, we truncate \eqref{eq:Kuser-full} to retain only the coefficients $S_{k,w}^{m,m,0}$ and restrict attention to a finite memory window $m\in[-M,M]$, yielding the working discrete-time channel model
\cite[eq.~(48)]{taghavi2006multiuser}, \cite[eq.~(8)]{dar2013properties},
\cite{ramachandran2023capacity}
\begin{align} \label{eq:Kuserapprox}
Y_{k}[i] &\approx X_{k}[i] \!+\! N_{k}[i]
\!+\! \jmath X_{k}[i] \!\!\! \sum_{m=-M}^{M} \!\!\! c_{k,w}^m
    \!\sum_{\substack{w=1\\w\ne k}}^{K} \!|X_w[i\!-\!m]|^{2}\!,
\end{align}
for $i=1,2,\dots,n,$ with $c_{k,w}^m := \gamma S_{k,w}^{m,m,0}$. In typical dispersion-uncompensated WDM systems, the dominant RP kernels $c_{k,w}^m$ associated with two-pulse XPM have been consistently observed to be non-negative under practical system parameters in prior analytical, numerical, and experimental studies \cite{dar2013properties,secondini2013achievable,dar2016pulse,ramachandran2023capacity}.
This non-negativity reflects the coherent phase shift induced by XPM, which is proportional to the instantaneous power of the interfering channels. For our simulation parameters in Table~\ref{table:param}, a numerical evaluation confirms $c_{k,w}^m \ge 0$ for all $w \neq k, m$, as illustrated in Fig.~\ref{fig:coeffs} ($M$ is taken as $10$ since the coefficients are negligible thereafter). 

\begin{table}[]
\def\arraystretch{1}%
\caption{Model Parameters}
\centering
\footnotesize
\begin{tabular}{c c}
\toprule
Parameter & Value \\ 
\midrule
Memory Length $M$ (eq. \eqref{eq:Kuserapprox}) & 10\\
Number of WDM users & 3-20 \\
Link Distance & $120 \: \textrm{km}$ \\
Nonlinearity parameter $\gamma$ & $1.2 \: \textrm{W}^{-1} \textrm{km}^{-1}$ \\
Signalling Rate & $64 \: \textrm{Gbaud}$ \\
Fiber attenuation $\alpha$ & $0.2 \: \textrm{dB/km}$ \\
Group velocity dispersion $\beta_2$ & $-21.7 \: \textrm{ps\textsuperscript{2}/km}$ \\
\bottomrule
\end{tabular}
\label{table:param} 
\end{table}
\normalsize
\begin{figure}[!b]
\vspace{-3cm}
\begin{center}
\scalebox{0.7}{
%
%
\definecolor{mycolor1}{rgb}{0.00000,0.44700,0.74100}%
\definecolor{mycolor2}{rgb}{0.85000,0.32500,0.09800}%
\begin{tikzpicture}

\begin{axis}[%
every axis/.append style={font=\large},
width=1.05\columnwidth,
height=0.8\columnwidth,
scale only axis,
xmin=0,
xmax=10,
xlabel={$m$},
ymin=0,
ymax=0.28,
ytick={0,0.05,0.1,0.15,0.2,0.25},
yticklabels={0,0.05,0.1,0.15,0.2,0.25},
ylabel={$\text{Coefficient} \:\: c_{k,w}^m$},
ylabel near ticks,
xlabel near ticks,
axis background/.style={fill=white},
legend style={
legend pos=north east,legend cell align=left,row sep=-0.5ex,grid style={dashed},font=\large
}
]

\addplot [color=mycolor2, mark = square*, mark size=2, line width=1.0pt]
  table[row sep=crcr, y expr=\thisrowno{1}*0.9]{%
0	0.269136\\
1	0.131544\\
2	0.103656000000001\\
3	0.0764399999999998\\
4	0.0624959999999994\\
5	0.0505680000000002\\
6	0.0376320000000003\\
7	0.0238560000000003\\
8	0.0119279999999993\\
9	0.00033600000000078\\
10	8.39999999993069e-05\\
};
\addlegendentry{$c_{1,2}^m$}

\addplot [color=teal, mark = triangle*, mark size=2.5, line width=1.0pt]
  table[row sep=crcr, y expr=\thisrowno{1}*0.9]{%
0	0.198240000000001\\
1	0.1055599999999992\\
2	0.0792400000000006\\
3	0.05846000000000006\\
4	0.0496399999999993\\
5	0.0389199999999996\\
6	0.02488\\
7	0.0134000000000001\\
8	0.00155999999999956\\
9	5.99999999995049e-05\\
10	5.99999999906231e-07\\
};
\addlegendentry{$c_{1,3}^m$}

\end{axis}

\begin{axis}[%
width=5.833in,
height=4.375in,
at={(0in,0in)},
scale only axis,
xmin=0,
xmax=1,
ymin=0,
ymax=1,
axis line style={draw=none},
ticks=none,
axis x line*=bottom,
axis y line*=left
]
\end{axis}
\end{tikzpicture}%
}
\end{center}
\caption{Channel coefficients $c_{1,2}^m, c_{1,3}^m$ involved in the model \eqref{eq:Kuserapprox} for user-$1$ being the channel of interest. Since the coefficients are symmetric about $m=0$, only the positive time indices are shown.}
\label{fig:coeffs}
\end{figure}
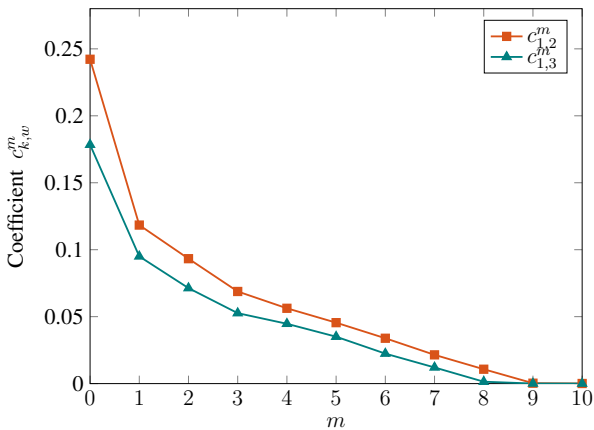

Model~\eqref{eq:Kuserapprox} has been widely used to describe XPM in multiuser optical systems and has been validated against split-step Fourier simulations over practical launch power ranges, while remaining analytically tractable
\cite{secondini2013achievable,dar2016pulse,ramachandran2023capacity}. As with all first-order RP models, it is not strictly energy preserving and may overestimate rates in the asymptotic high-power regime \cite{keykhosravi2019accuracy,vannucci2002rp}; this scope is discussed further in Remark~\ref{rem:prelog}.

\section{Main Results}
Before presenting the main results, we briefly review the standard \emph{treat-interference-as-noise} (TIN) baseline for \eqref{eq:Kuserapprox}. Under TIN, user~$k$ decodes only its own message, treating the signals of all other users as random interference. The achievable rate is thus given by the mutual information between the input and output of user~$k$ under the induced channel law~\cite[eq.~(6.2)]{el2011network}. The resulting rates typically exhibit a pronounced peak, often referred to as the \emph{nonlinear Shannon limit}~\cite{mitra2001nonlinear}.

In addition to the TIN baseline, Section~V also considers a successive decoding--equalization (SDE) benchmark tailored to the RP channel model. This scheme enables simultaneous transmission by successively decoding users and equalizing the corresponding nonlinear distortion, and serves as a multi-user comparison to our inner bounds proposed in the sequel.

\subsection{Capacity Region Outer Bounds and Single-user Capacity in the limit of high SNR} \label{sec:Kuser}
In this section, we characterize the set of achievable rate tuples $(R_1,R_2,\ldots,R_K)$ for the $K$-user channel shown in Fig.~\ref{fig:ICmodel} and represented by~\eqref{eq:Kuserapprox}. Before formalizing our capacity results, we first introduce some basic information-theoretic definitions that will be useful in the sequel.  

An $(n,2^{nR_1},2^{nR_2},\ldots,2^{nR_K})$ code for this setting is specified by the following components:  
\begin{itemize}
    \item $K$ message sets $\{1,2,\ldots,2^{nR_k}\} \triangleq \{1:2^{nR_k}\}$ for each $k \in \{1,2,\ldots,K\} \triangleq \{1:K\}$,  
    \item $K$ encoders $E_k$ mapping a message $M_k \in \{1:2^{nR_k}\}$ to a codeword $X_k^n(M_k)$, subject to power constraints $\frac{1}{n} \sum_{i=1}^n \mathbb{E}\!\left[|X_{k}[i]|^2\right] \leq P_k$ for each $k \in \{1:K\}$, and  
    \item $K$ decoders $D_k$ producing an estimate $\hat{M}_k$ from the received sequence $Y_k^n$ for each $k \in \{1:K\}$.  
\end{itemize}

Each message $M_k$ is assumed to be independently and uniformly distributed over its alphabet $\{1,2,\ldots,2^{nR_k}\}$. The overall error probability is given by  
\begin{align} \label{eq:errorprob}
P_e \triangleq \Pr\big( (\hat{M}_1(Y_1^n),\ldots,\hat{M}_K(Y_K^n)) \neq (M_1,\ldots,M_K) \big).
\end{align}

A rate vector $(R_1,R_2,\ldots,R_K)$ is called \emph{achievable} if there exists a sequence of $(n,2^{nR_1},2^{nR_2},\ldots,2^{nR_K})$ codes for which $\lim_{n \to \infty} P_e = 0$. The $K$-user capacity region $\mathcal{C}_{K\textrm{-cap-reg}}(P_1,P_2,\ldots,P_K)$ is then defined as the closure of the set of all achievable rate tuples (the explicit power dependence of $\mathcal{C}_{K\textrm{-cap-reg}}$ will be suppressed in the sequel for brevity). Finally, the single-user capacity of user-$k$ is  
\[
C_k \triangleq \sup \{R_k \;|\; (R_1,R_2,\ldots,R_K) \in \mathcal{C}_{K\textrm{-cap-reg}}\}.
\]
\noindent \paragraph*{High-SNR regime}
We say the channel operates in the \emph{high--SNR regime} if the transmit powers satisfy
\[
\min_{1\le k\le K} P_k \;\longrightarrow\; \infty .
\] Equivalently, this corresponds to the regime in which the blockwise lower envelope of the instantaneous signal-to-noise ratios $\gamma_{k,i} = \frac{\mathbb{E}[|X_k[i]|^2]}{2\sigma_k^2}
\Big( 1 + \bigl(\sum_{m}  \sum_{w\neq k} c_{k,w}^m |x_w[i-m]|^2\bigr)^2 \Big)$ satisfies $\inf_i \gamma_{k,i} \longrightarrow \infty \quad \text{for all } k$.
All limits and extremal characterizations in the current Section~\ref{sec:Kuser} are taken with respect to this scaling. Let $\mathcal{C}_{K\textrm{-cap-reg}}(P_1,P_2,\dots,P_K)$ denote the $K$--user capacity region for given average powers $\mathbf P=(P_1,P_2,\dots,P_K)$. The \emph{high--SNR capacity region} is the Kuratowski set limit (see \cite[Ch.~4]{rockafellar1998variational}) as the average powers grow without bound:  
\begin{align}
\overline{\mathcal{C}}_{K\textrm{-cap-reg}}^{\mathrm{high\!-\!SNR}}
&\!:=\! \limsup_{\min_k P_k\to\infty} \mathcal{C}_{K\textrm{-cap-reg}}(P_1,P_2,\dots,P_K) \notag\\
&\!=\! \bigcap_{P_0>0}\ \overline{\bigcup_{\min_k P_k\ge P_0}\mathcal{C}_{K\textrm{-cap-reg}}(P_1,P_2,\dots,P_K)},\\
\underline{\mathcal{C}}_{K\textrm{-cap-reg}}^{\mathrm{high\!-\!SNR}}
&\!:=\! \liminf_{\min_k P_k\to\infty} \mathcal{C}_{K\textrm{-cap-reg}}(P_1,P_2,\dots,P_K)  \notag\\
&\!=\! \bigcup_{P_0>0}\ \overline{\bigcap_{\min_k P_k\ge P_0}\mathcal{C}_{K\textrm{-cap-reg}}(P_1,P_2,\dots,P_K)}.
\end{align}
Here the notation $\overline{S}$ stands for the Euclidean closure of a set $S$.  
Under the high-SNR scaling $P_k \to \infty$ for all $k$, $\{\mathcal{C}_{K\textrm{-cap-reg}}(\mathbf P)\}$ is monotone with respect to the operation of set inclusion, in the sense that
\[
\mathbf P \preceq \mathbf P' \;\; \Rightarrow \;\;
\mathcal{C}_{K\textrm{-cap-reg}}(\mathbf P) \subseteq \mathcal{C}_{K\textrm{-cap-reg}}(\mathbf P'),
\]
where $\preceq$ denotes component-wise ordering. As a result, the Kuratowski $\liminf$ and $\limsup$ of $\mathcal{C}_{K-\textrm{cap-reg}}(\mathbf P)$ coincide and the setwise limit exists. When these two limits coincide, we denote their common value as follows:
\[
\mathcal{C}_{K\textrm{-cap-reg}}^{\mathrm{high\!-\!SNR}}
:= \overline{\mathcal{C}}_{K\textrm{-cap-reg}}^{\mathrm{high\!-\!SNR}}
= \underline{\mathcal{C}}_{K\textrm{-cap-reg}}^{\mathrm{high\!-\!SNR}} .
\]
The \emph{high--SNR single--user capacity} for user $k$ is then
\[
C_k^{\mathrm{high\!-\!SNR}}
\;\triangleq\;
\sup \bigl\{ R_k \,\big|\, (R_1,R_2,\dots,R_K)\in
\mathcal{C}_{K\textrm{-cap-reg}}^{\mathrm{high\!-\!SNR}} \bigr\}.
\]

\begin{theorem} \label{thm:OBgenieK}
The single-user capacity $C_k$ of the interference channel in \eqref{eq:Kuserapprox} at high SNR is characterized by
\begin{align} \label{eq:capLBUB}
C_k^{\mathrm{high\!-\!SNR}} &= \log\biggl(\frac{P_{k}}{2\sigma_k^2}\biggl[1+\biggl(\sum_{\substack{w=1\\w \neq k}}^K P_w \sum_{m=-M}^M c_{k,w}^m\biggr)^2\biggr]\! \biggr),
\end{align}
for all $k \in \{1,2,\ldots,K\}$. The collection of the single-user high-SNR upper bounds on $C_k$ in \eqref{eq:capLBUB} for all $k \in \{1,2,\ldots,K\}$ constitute the following outer bound:
\begin{align} \label{eq:OBregion}
\mathcal{C}_{K\textrm{-cap-reg}}^{\mathrm{high\!-\!SNR}} 
   &\subseteq 
   \bigg\{ (R_1,R_2,\ldots,R_K) \;\bigg|\;
   R_k \leq C_k^{\mathrm{high\!-\!SNR}}, \notag\\
   &\phantom{wwwwwwwwwwwww}
   \;\forall k \in \{1,2,\ldots,K\}
   \bigg\}.
\end{align}
\end{theorem}
\begin{proof}
See Section~\ref{app:genieK}.
\end{proof}
In addition to providing an upper bound on the rate of user~$k$, Theorem~\ref{thm:OBgenieK} also characterizes the extremal interferer signaling (viz. users $w\neq k$) that maximizes this rate. Specifically, Section~\ref{app:genieK} shows that constant-amplitude signaling by users $w\neq k$ is optimal. This characterization naturally motivates the achievable strategies discussed in Section~\ref{sec:ach}.

\begin{remark} 
\textbf{Interpretation of Achievability and Input Distributions:} Throughout this paper, the term \emph{achievable} refers to rates attained within the first-order RP channel model under average power constraints. As will be clear from the proof of Theorem~\ref{thm:OBgenieK}, Gaussian inputs arise as the optimal choice for the user of interest since the genie-aided outer bound establishes that constant-amplitude interferer signaling is extremal. Under this provably optimal interferer structure, the effective channel reduces to a linear additive-noise model within the RP approximation, for which Gaussian inputs are optimal in the standard information-theoretic sense. The resulting rate expressions thus characterize the single-user capacity corner points of the RP channel under average power constraints. 
\end{remark}

\vspace{-5mm}

\subsection{Beyond high-SNR} \label{sec:Kuserfinite}
Theorem~\ref{thm:OBgenieK} provides a high–SNR characterization without placing any sign constraints on the perturbative coefficients. Empirically however, the dominant RP coefficients $c_{k,w}^m$ in \eqref{eq:Kuserapprox} are always observed to be nonnegative for typical WDM configurations, which is extensively validated in e.g.,~\cite{dar2013properties,secondini2013achievable,ramachandran2023capacity}.
This empirically validated condition can be leveraged to strengthen our analysis and obtain a \emph{finite–SNR} extremality result that leads to the following capacity characterization.

\begin{prop}[Capacity characterization under non-negative perturbative kernels] \label{thm:OBgenieKfinite}
The single-user capacity $C_k$ of the interference channel \eqref{eq:Kuserapprox} under $c_{k,w}^m \geq 0$ is characterized by\footnote{This characterization holds for the truncated first‑order RP model \eqref{eq:Kuserapprox} under the condition $c_{k,w}^m \geq 0$ for all $w \neq k$, $m \in [-M,M]$. It is not claimed to hold for the full NLSE channel or when higher‑order RP terms are included.}
\begin{align} \label{eq:capLBUBfinite}
C_k &= \log\biggl(1+\frac{P_{k}}{2\sigma_k^2}\biggl[1+\biggl(\sum_{\substack{w=1\\w \neq k}}^K P_w \sum_{m=-M}^M c_{k,w}^m\biggr)^2\biggr] \biggr),
\end{align}
for all $k \in \{1,2,\ldots,K\}$. The collection of the single-user upper bounds on $C_k$ in \eqref{eq:capLBUBfinite} for all $k \in \{1,2,\ldots,K\}$ constitute the following outer bound:
\begin{align} \label{eq:OBregionfinite}
\mathcal{C}_{K\textrm{-cap-reg}} 
   &\subseteq 
   \bigg\{\! (R_1,R_2,\cdot\cdot,R_K) \bigg|
   R_k \leq C_k \:\forall k \in \{1\!:\!K\}
   \!\!\bigg\}\!.
\end{align}
\end{prop}
\begin{proof}
See Section~\ref{app:genieKfinite}.
\end{proof}

Equation~\eqref{eq:capLBUB} follows from \eqref{eq:capLBUBfinite} by the standard high-SNR approximation $\log(1+x)\approx\log x$, which removes the additive constant inside the logarithm.

\begin{remark}[Pre-log growth under the RP model] \label{rem:prelog}
In the symmetric scaling regime $P_1\asymp P_2\asymp\cdots\asymp P_K\asymp P$, the expression
in~\eqref{eq:capLBUBfinite} (or \eqref{eq:capLBUB}) grows as
\[
C_k
\sim \log\!\bigl(P_k P^2\bigr)
=3\log P+O(1),
\]
so the capacity pre-log can reach $3$ under the first-order RP model. Such superlinear pre-log factors have been reported previously for dispersionless single-channel RP models (e.g., \cite[Sec.~VI]{vannucci2002rp},\cite[Sec.~3.1 and 5]{keykhosravi2019accuracy}) and reflect the fact that first-order RP is not strictly energy preserving, leading to an artificial power gain when all channels scale simultaneously. Thus, our capacity characterizations should be interpreted within the RP approximation; more complete models, including signal-noise interactions and higher-order terms, are expected to moderate the asymptotic slope. Nevertheless, our numerical results in Section~\ref{sec:num} show that the derived bounds closely track achievable rates over the validated $0$-$5$ dBm regime.
\end{remark}

We emphasize that the effective gain term in~\eqref{eq:Kuserapprox} due to coherent XPM should not be interpreted as physical signal amplification. This model has also been validated against split-step Fourier simulations and experiments over practical launch-powers, capturing dominant XPM effects and performance trends \cite{mecozzi2002limits,johannisson2014modeling,keykhosravi2019accuracy,dar2013properties}. Accordingly, our results are best viewed as characterizing extremal signaling and interference-management mechanisms within the RP framework. 


\subsection{All-helps-one (AHO) Inner Bound} \label{sec:ach}
We have the following achievable region (inner bound) to the capacity region. Analogous results can be stated for the high-SNR regime, which we do not repeat here.
\begin{prop}[All-Helps-One Inner Bound] \label{prop:allhelpsone}
Consider the $K$-user interference channel described in~\eqref{eq:Kuserapprox}.  
\noindent For each $k$, let $\mathbf{r}^{(k)}=(r^{(k)}_1,r^{(k)}_2,\ldots,r^{(k)}_K)\in\mathbb{R}_+^K$ be the rate vector with
\[
r^{(k)}_j =
\begin{cases}
C_k, & j=k,\\[4pt]
0, & j\neq k.
\end{cases}
\]

\noindent Define the \emph{all-helps-one} region as the convex hull of these $K$ extreme points:
\begin{align} \label{eq:AHOregion}
\mathcal{R}_{\text{AHO}}
&\triangleq \mathrm{Conv}\big\{\mathbf{r}^{(1)},\mathbf{r}^{(2)},\ldots,\mathbf{r}^{(K)}\big\} \\
&= \Big\{ \mathbf{R}\in\mathbb{R}_+^K \;\Big|\; \mathbf{R}=\sum_{k=1}^K \lambda_k \mathbf{r}^{(k)},\;
\lambda_k\ge 0,\; \sum_{k=1}^K \lambda_k = 1 \Big\}. \nonumber
\end{align}

\noindent Then the region $\mathcal{R}_{\text{AHO}}$ is achievable, i.e.,
\begin{align}
\mathcal{C}_{K\textrm{-cap-reg}} \supseteq \mathcal{R}_{\text{AHO}}.
\end{align}

\noindent Equivalently, letting $\lambda_k$ denote the fraction of time during which user-$k$ is the ``helped'' user, the region can be written coordinate-wise as
\begin{align*}
R_k \le \lambda_k \, C_k, \quad k=1,2,\ldots,K,\quad
\lambda_k\ge 0,\;\sum_{k=1}^K \lambda_k = 1.
\end{align*}
\end{prop}

\begin{remark}
AHO should be interpreted as an information-theoretic characterization of long-term rate trade-offs, rather than a specific transmission protocol. In practice, the associated time sharing may occur over longer timescales (than symbol-level coordination), such as transmission epochs, wavelength reconfiguration, or network-level scheduling. 
\end{remark}


\begin{proof}
The achievability proof for $\mathcal{R}_{\text{AHO}}$ proceeds as follows. Fix $k \in \{1,2,\ldots,K\}$ and suppose that all interferer transmissions are chosen as\footnote{The constant-amplitude interferer assignment in \eqref{eq:constamp} cannot be strictly enforced for the first and last $M$ symbols, incurring an $O(M/n)$ rate penalty that vanishes as $n\to\infty$, so the expression \eqref{eq:capLBUBfinite} holds as a capacity result.}
\begin{align}
|x_{w}[i-m]|^2 = P_w, 
\quad \forall \; w \neq k, \; \forall \; m, \; \text{and} \; \forall \; i > m, \label{eq:constamp}
\end{align}
i.e., users $w \neq k$ employ constant-amplitude signaling. This choice turns out to be optimal in terms of maximizing the achievable rate of user-$k$, since it precisely achieves the single-user upper bound \eqref{eq:capLBUBfinite} in Proposition~\ref{thm:OBgenieKfinite}. Indeed, under this condition, the channel reduces to the following memoryless single-user model (with time indices omitted for brevity):  
\begin{align}
Y_{k} &= X_{k}\biggl(1+j \sum_{\substack{w=1\\w \neq k}}^K P_w \sum_{m=-M}^M c_{k,w}^m \biggr)+N_{k}, \label{eq:shamai}
\end{align}
which is a complex AWGN channel with average input power constraint $P_k$, whose capacity is well known. Hence, the achievable rate is given by \eqref{eq:capLBUBfinite}. Repeating this construction for each user yields $K$ corner points of the achievable rate region. Their convex hull, obtained by time sharing, defines the \emph{all-helps-one} inner bound for \eqref{eq:Kuserapprox}. This construction should be viewed as an information-theoretic characterization of trade-offs under non-cooperating users, rather than a simultaneous multi-user coding scheme. A numerical comparison with TIN and the outer bound of Prop.~\ref{thm:OBgenieKfinite} is provided in Section~\ref{sec:num}.
\end{proof}

\begin{remark}
The analysis above employs continuous-valued (Gaussian) inputs for analytical tractability within the first-order RP model, whereas practical optical systems use discrete modulation formats. While \eqref{eq:capLBUBfinite} pertains to continuous signaling, the key qualitative conclusions, namely the extremality of constant-amplitude interferer signaling and its advantage over TIN, are expected to extend to sufficiently rich discrete constellations with similar amplitude statistics. 
\end{remark}

\begin{remark}[Power Accounting and TDMA Comparison]
Under symmetric operation with $K$ users, each user acts as the active user for a fraction $1/K$ of the time. Consequently, the throughput of the AHO scheme is the average symmetric rate $C_k/K$, where the factor $1/K$ accounts for the power expended while a user serves as a helper. Thus, the helper power is reflected through the averaging over time (see Fig.~\ref{fig:succ_vs_tin_vs_aho}).

For comparison, conventional time-division multiple access (TDMA) silences all interfering users, rendering the channel linear and yielding a symmetric rate of \smash{$\frac{1}{K}\log\!\left(1+\frac{P_k}{2\sigma_k^2}\right)$}. In contrast, AHO employs constant-amplitude helper signals whose XPM contributions increase the effective SNR of the active user, resulting in a symmetric throughput of $C_k/K$.
\end{remark}

\section{Proof of Theorem \ref{thm:OBgenieK}} \label{app:genieK}
We now establish the outer bound proof. 
The rate of user$-k$, $k \in \{1,2,\cdots,K\}$, can be upper bounded as follows:
\begin{align}
&nR_k = H(M_k) \notag\\
&\stackrel{(a)}= H(M_k|\{X_w^n | w \neq k\}) \notag\\
&= H(M_k|\{X_w^n | w \neq k\})-H(M_k|\{X_w^n | w \neq k\},Y_k^n) \notag\\
&\phantom{wwww}+H(M_k|\{X_w^n | w \neq k\},Y_k^n) \notag\\
&\stackrel{(b)}\leq I(M_k;Y_k^n|\{X_w^n | w \neq k\})+H(M_k|Y_k^n) \notag\\
&\stackrel{(c)}\leq I(M_k;Y_k^n|\{X_w^n | w \neq k\})+1+P_{e} nR_k  \notag\\
&\stackrel{(d)}= I(M_k;Y_k^n|\{X_w^n | w \neq k\})+n\epsilon_n  \notag\\
&\stackrel{(e)}\leq I(X_k^n;Y_k^n|\{X_w^n | w \neq k\})\!+\!n\epsilon_n \notag\\
&= h(Y_k^n|\{X_w^n | w \neq k\})-\!h(Y_k^n|X_k^n,\{X_w^n | w \neq k\})+\!n\epsilon_n \notag\\
&= h(Y_k^n|\{X_w^n | w \neq k\})-\!h({N_k}^n)+\!n\epsilon_n \notag\\
&\stackrel{(f)}\leq  \max_{\{x_{w}^n | w \neq k\}} \left[ h(Y_{k}^n|\{X_w^n=x_{w}^n | w \neq k\})\right] - h(N_{k}^n)+n\epsilon_n \notag\\
&\stackrel{(g)}\leq  \max_{\{x_{w}^n | w \neq k\}} \log\left(\textup{det}\left(\textup{cov}\left(Y_k^n|\{X_w^n = x_w^n | w \neq k\}\right)\right)\right) \notag\\
&\phantom{wwwwwww}-\log\left(\textup{det}\left(\textup{cov}\left(N_k^n\right)\right)\right)+n\epsilon_n \notag\\
&\stackrel{(h)}=  \max_{\{x_{w}^n | w \neq k\}} \log\left(\textup{det}\left(\textup{cov}\left(Y_k^n|\{X_w^n = x_w^n | w \neq k\}\right)\right)\right) \notag\\
&\phantom{wwwwwww}-n\log\left(2\sigma_k^2\right)+n\epsilon_n, \label{eq:convKg}
\end{align}
where (a) follows from the independence of $M_k$ and $\{X_w^n:w\neq k\}$, (b) from the fact that conditioning reduces entropy, and (c) from Fano’s inequality with error probability $P_e$ defined in \eqref{eq:errorprob}. Step (d) introduces $\epsilon_n := \tfrac{1}{n}+P_e R_k$, which satisfies $\epsilon_n\to 0$ as $n\to\infty$, and (e) follows from the data processing inequality, since $M_k \to X_k^n \to Y_k^n$ forms a Markov chain conditioned on $\{X_w^n:w\neq k\}$, while (f) uses the fact that $h(Y_k^n \mid \{X_w^n:w\neq k\})$ is an average over realizations $\{x_w^n:w\neq k\}$ and is therefore upper bounded by the maximum. Step (g) follows since Gaussian random vectors maximize differential entropy for a fixed covariance, and (h) uses that the additive noise is i.i.d. circularly symmetric complex Gaussian with covariance $2\sigma_k^2 I$.

In expression~\eqref{eq:convKg}, it remains to compute $\textup{det}\left(\textup{cov}\left(Y_k^n|\{X_w^n = x_w^n | w \neq k\}\right)\right)$. Towards this end, we arrange expression~\eqref{eq:Kuserapprox} conditioned on $\{X_w^n = x_w^n | w \neq k\}$ in vector form to obtain
\begin{align}
Y_k^{n}\big|\{X_w^n = x_w^n | w \neq k\} = \left(I+\Theta\right) X_k^{n} +N_k^{n}, \label{eq:vecform}
\end{align}
where $\Theta$ is an $n \times n$ diagonal matrix.
In particular, the diagonal elements of $\Theta$ are nothing but
\begin{align}
&\theta(i,i) = j \sum_{\substack{w=1\\w \neq k}}^K\:\: \sum_{m=-M}^{M} c_{k,w}^m |x_{w}[i-m]|^2.
\end{align}
Hence, using \eqref{eq:vecform}, the determinant of interest becomes
\begin{align}
&\textup{det}\left(\textup{cov}\left(Y_k^n|\{X_w^n = x_w^n | w \neq k\}\right)\right) \notag\\
&\phantom{www}= \textup{det}\left(\left(I+\Theta\right) \textup{cov}\left(X_k^{n}\right) \left(I+\Theta\right)^H+ 2\sigma_k^2 I  \right) \notag\\
&\phantom{www}\stackrel{(a)}= \textup{det}\left(\left(I+\Theta\right)^H \left(I+\Theta\right) \textup{cov}\left(X_k^{n}\right) + 2\sigma_k^2 I  \right), \label{eq:iandtheta}
\end{align}
where (a) follows from Sylvester’s determinant identity. Let $P_{k,i}:=\mathbb{E}[|X_k[i]|^2]$ denote the diagonal entries of $\operatorname{cov}(X_k^n)$, so that the average power constraint $\sum_{i=1}^n P_{k,i}\le nP_k$ imposes a trace constraint on $\operatorname{cov}(X_k^n)$. Since $(I+\Theta)^H(I+\Theta)$ in \eqref{eq:iandtheta} is diagonal, this trace constraint implies that the optimal input covariance is diagonal; see Lemma~\ref{lem:diag-opt} in Appendix~\ref{app:diag}. Consequently, expression \eqref{eq:iandtheta} becomes
\begin{align} \label{eq:detK}
&\textup{det}\left(\textup{cov}\left(Y_k^n|\{X_w^n = x_w^n | w \neq k\}\right)\right) \notag\\
&= \textup{det}\left(\left(I+\Theta\right)^H \left(I+\Theta\right) \textup{cov}\left(X_k^{n}\right) + 2\sigma_k^2 I  \right) \notag\\
&= \prod_{i=1}^n \biggl(\!P_{k,i}\biggl(\!1+\biggl(\sum_{\substack{w=1\\w \neq k}}^K \sum_{m=-M}^{M} c_{k,w}^m |x_{w}[i-m]|^2 \!\biggr)^2\biggr)+2\sigma_k^2\!\biggr)\!.
\end{align}
From \eqref{eq:convKg} and \eqref{eq:detK}, we obtain the following expression for the rate of user-$k$, $k \in \{1,2,\cdots,K\}$ (suppressing the limits of the summations over $w$ and $m$):
\begin{align}
&nR_k \notag\\
&\leq \max_{\{x_w^n | w \neq k\}}\! \sum_{i=1}^n \log\!\biggl(\!1\!+\!\frac{P_{k,i}}{2 \sigma_k^2}\!\biggl(\!1\!+\!\!\biggl(\sum_{w \neq k} \sum_{m} c_{k,w}^m |x_{w}[i\!-\!m]|^2 \!\!\biggr)^2\biggr)\!\!\!\biggr) \notag\\
&\phantom{wwwwwwwww}+n\epsilon_n. \label{eq:optim}
\end{align}
The above expression involves two optimizations:
\begin{enumerate}
\item the power allocation $\{P_{k,i}\}_{i=1}^n$ for the user of interest, subject to $\sum_{i=1}^n P_{k,i}\le nP_k$, and
\item the interferer transmissions $\{x_w^n:w\neq k\}$, each subject to $\sum_{i=1}^n |x_w[i]|^2\le nP_w$.
\end{enumerate}
We address these optimizations separately: for general coefficients $c_{k,w}^m$ in the high-SNR regime (next subsection), and for non-negative coefficients $c_{k,w}^m\ge 0$ at finite SNR
(Section~\ref{app:genieKfinite}, ahead).

\subsection{Extremal structure of the maximization}
The maximization appearing in~\eqref{eq:optim} is, in general, analytically intractable. 
Even though each interferer’s feasible set
\(
\mathcal{F}_w := \bigl\{x_w^n \in \mathbb{C}^n : \tfrac{1}{n}\sum_{i=1}^n |x_w[i]|^2 \le P_w \bigr\}
\)
is convex and the joint constraint set for all $w\ne k$ is a Cartesian product of such convex sets, the objective
\begin{equation*}
(x_w^n)_{w\ne k} \! \mapsto \! \sum_{i=1}^n 
   \log\!\biggl(\!1+\frac{P_{k,i}}{2\sigma_k^2}
       \biggl[\!1+\biggl(\!\sum_{w\ne k}\sum_{m} c_{k,w}^m\,|x_w[i-m]|^2\!\!\biggr)^{2}\!\biggr]\!\biggr)
\end{equation*}
is not jointly concave in the interferer sequences $(x_w^n)_{w\ne k}$. Consequently, there is no known closed-form solution or a general efficient algorithm for this global maximization under average power constraints.

We therefore restrict attention to the high-SNR regime, where $\log(1+x)\sim\log x$ isolates the leading term and renders the optimization tractable. In this regime, the constant term inside the logarithm is negligible compared to the signal-to-noise ratio $\frac{P_{k,i}}{2\sigma_k^2}$. Formally, by Lemma~\ref{lem:log_approx} in Appendix~\ref{sec:appA}, with $\gamma_{k,i} =\frac{P_{k,i}}{2\sigma_k^2}\Bigl[1+
\Bigl(\sum_{w\neq k}\sum_{m=-M}^M c_{k,w}^m |x_w[i-m]|^2\Bigr)^2\Bigr],$ the per-use mutual information satisfies $\frac{1}{n}\sum_{i=1}^n \log(1+\gamma_{k,i}) =\frac{1}{n}\sum_{i=1}^n \log\gamma_{k,i}+o(1)$ as the powers grow, uniformly over admissible strategies. Hence, we can write the following chain of inequalities (in the regime where user-$k$ has high SNR)
\begin{align}
&n(R_k-\epsilon_n) \notag\\
&\lessapprox \max_{\{x_w^n | w \neq k\}}\! \sum_{i=1}^n \log\biggl(\!\frac{P_{k,i}}{2 \sigma_k^2}\!\biggl(\!1\!+\!\biggl(\sum_{w \neq k} \sum_{m} c_{k,w}^m |x_{w}[i\!-\!m]|^2 \!\!\biggr)^2\biggr)\!\!\!\biggr) \notag\\
&= \max_{\{x_w^n | w \neq k\}} \sum_{i=1}^n \log\biggl(\frac{P_{k,i}}{2 \sigma_k^2}\biggr) \notag\\
&\phantom{www}+\max_{\{\mathbf{x}_w | w \neq k\}}\! \sum_{i=1}^n \log\biggl(\!1\!+\!\biggl(\sum_{w \neq k} \sum_{m} c_{k,w}^m |x_{w}[i\!-\!m]|^2 \!\biggr)^2\biggr) \notag\\
&\stackrel{(a)}\leq n \log \biggl(\frac{1}{n}\sum_{i=1}^n \frac{P_{k,i}}{2 \sigma_k^2}\biggr) \notag\\
&\phantom{w}+\max_{\{x_w^n | w \neq k\}}\! n\log\biggl(\!1\!+\!\frac{1}{n}\sum_{i=1}^n\biggl(\sum_{w \neq k} \sum_{m} c_{k,w}^m |x_{w}[i\!-\!m]|^2 \!\biggr)^2\biggr) \notag\\
&\stackrel{(b)}\leq n \log \biggl(\frac{P_{k}}{2 \sigma_k^2}\biggr) \notag\\
&\phantom{w}+\max_{\{x_w^n | w \neq k\}}\! n\log\biggl(\!1\!+\!\frac{1}{n}\sum_{i=1}^n\biggl(\sum_{w \neq k} \sum_{m} c_{k,w}^m |x_{w}[i\!-\!m]|^2 \!\biggr)^2\biggr)\!, \label{eq:optimhighSNR}
\end{align}
where (a) follows from Jensen's inequality\footnote{By Lemma 2, the high‑SNR approximation $\log(1+\gamma_{k,i}) = \log\gamma_{k,i} + o(1)$ decouples the power allocation $\{P_{k,i}\}$ from the interferer sequences $\{x_w^n\}$, allowing the maximizations to be separated since $P_{k,i}$ appears in $\gamma_{k,i}$ only via the first term of the numerator.} and (b) follows from the power constraint on user-$k$.

Since the logarithm is monotone, it suffices to maximize
\begin{align}
M
=\max_{\{x_w^n:w\neq k\}} \sum_{i=1}^n \Biggl(\sum_{\substack{w=1\\w\neq k}}^{K} \sum_{m=-M}^{M}
c_{k,w}^m |x_w[i-m]|^2 \Biggr)^2,
\end{align}
subject to $\sum_{i=1}^n |x_w[i]|^2 \le nP_w$ for all $w\neq k$.

Define $z_w[i]:=|x_w[i]|^2$, $\underline{z}_w=(z_w[1],z_w[2],\ldots,z_w[n])^\top$, and $\underline{z}:=\{\underline{z}_w\}_{w\neq k}$. Then 
\begin{align}
M = \max_{\underline{z}\ge 0} \;\underline{z}^\top C\,\underline{z} \triangleq \max_{\underline{z}\ge 0} Q(\underline z),
\end{align}
where $C$ is an $n(K\!-\!1)\times n(K\!-\!1)$ matrix determined by
$\{c_{k,w}^m\}$. The constraints become ($\mathbf{1}$ being the all-ones vector)
\begin{align}
\mathbf{1}^\top \underline{z}_w \le nP_w, \qquad \forall\, w\neq k .
\end{align}

The structure of $C$ is as follows. For each interfering user $w\neq k$, define the $n\times n$ Toeplitz matrix
\[
(T_{k,w})_{i,j} =
\begin{cases}
c_{k,w}^{\,i-j}, & |i-j|\le M,\\
0, & \text{otherwise},
\end{cases}
\qquad i,j\in\{1,2,\ldots,n\}.
\]
Let
\[
L_k \triangleq \big[\,T_{k,w}\,\big]_{w\neq k}\in\mathbb{R}^{\,n\times n(K-1)}.
\]
Then
\[
(L_k\underline z)_i=\sum_{w\neq k}\sum_{m=-M}^{M}c_{k,w}^{\,m}z_w[i-m],
\]
and the objective can be written as
\[
Q(\underline z)=\sum_{i=1}^n\big((L_k\underline z)_i\big)^2
=\|L_k\underline z\|_2^2
=\underline z^{\mathsf T}(L_k^\top L_k)\underline z .
\]
Hence,
\[
C=L_k^\top L_k\succeq 0 .
\]

We therefore maximize the convex quadratic $Q(\underline z)=\underline z^{\mathsf T}C\underline z$ subject to the average power constraints $\mathbf 1^\top\underline z_w\le nP_w$ for all $w\neq k$ and $\underline z\ge 0$. Since a convex function is maximized over a convex set at an extreme point, any global maximizer must lie on an
extreme point of the feasible set, as formalized next.


\begin{prop}[Extremal structure under average power]\label{prop:extreme}
Fix user $k$ and blocklength $n$. For each interferer $w\neq k$,
let
\[
\mathcal{F}_w :=\Bigl\{\underline z_w\in\mathbb{R}_+^{\,n}: \tfrac{1}{n}\mathbf 1^\top\underline z_w=P_w\Bigr\},
\qquad \mathcal{F}:=\prod_{w\neq k}\mathcal{F}_w ,
\]
and define $\underline z:=(\underline z_w)_{w\neq k}$. Consider the quadratic form
\[
Q(\underline z):=\underline z^{\mathsf T}C\underline z,
\]
where $C\succeq 0$ is the symmetric matrix shown above. Then:
\begin{enumerate}
\item $Q$ is convex on $\mathcal{F}$.
\item Any maximizer of $Q$ over $\mathcal{F}$ can be chosen at an extreme point of $\mathcal{F}$, i.e., for each $w\neq k$ there exists $j_w\in\{1,2,\ldots,n\}$ such that
\[
\underline z_w = nP_w\,\mathbf e_{j_w},
\]
\end{enumerate}
where \(\mathbf{e}_j\) is the \(j\)-th standard basis vector of \(\mathbb{R}^n\). In particular, $Q$ achieves its supremum over $\mathcal{F}$ at such an extreme allocation.
\end{prop}

\begin{proof}
Please refer to Appendix~\ref{proof:prop3}.
\end{proof}

In particular, note that the time-constant choice $\underline{z_w} = P_w\mathbf{1}$ for each $w\neq k$ belongs to $\mathcal{F}$ because
\(
P_w\mathbf{1} = \frac{1}{n}\sum_{j=1}^n nP_w\: \mathbf{e}_j.
\)
Hence, the constant-amplitude assignment is one particular feasible point
that attains the same $Q(\underline{z})$ value as any other choice with the same per-user averages.

Thus, Proposition~\ref{prop:extreme} implies that the outer bound from \eqref{eq:optimhighSNR} can finally be rewritten as
\begin{align*}
nR_k \leq n\log\biggl(\!\frac{P_k}{2\sigma_k^2} \biggl[1\!+\!\biggl(\sum_{\substack{w=1\\w \neq k}}^K P_w \sum_{m=-M}^M c_{k,w}^m \biggl)^2\biggl] \biggl)+n\epsilon_n,
\end{align*}
where the last step follows by Jensen's inequality, with $\epsilon_n \to 0$ as $n \to \infty$. Clearly, this upper bound is achievable by constant amplitude signalling for all the interferers $w \neq k$, along with a Gaussian codebook for user-$k$. Hence, we conclude that the high-SNR capacity of user-$k$ is indeed characterized by \eqref{eq:capLBUB}.

\begin{remark}
We emphasize that the genie-aided upper bound holds for arbitrary input distributions satisfying the average power constraints. The derivation relies only on upper-bounding the conditional differential entropy of the channel output by that of a Gaussian random vector with the same covariance, and therefore imposes no structural assumptions on the inputs. Assumptions such as constant-amplitude signaling arise only in the achievability analysis and in identifying signaling strategies that attain the bound.
\end{remark}

\vspace{-4mm}

\subsection{Extremality under nonnegative kernels} \label{app:genieKfinite}
In the empirically validated case of the channel coefficients $c_{k,w}^m$ being non-negative~\cite{dar2013properties,secondini2013achievable,ramachandran2023capacity}, we can go further and obtain the finite-SNR capacity characterization as well.
\begin{prop}[Finite–SNR extremality under nonnegative kernels]\label{prop:finiteSNR-extreme}
Fix a user of interest $k$, and suppose the retained perturbative coefficients are nonnegative:
\[
c_{k,w}^m = \gamma S_{k,w}^{m,m,0} \;\ge\; 0
\qquad\forall\,w\ne k,\; m\in[-M,M].
\]
Then, for any fixed average powers $\{P_w\}_{w\ne k}$, the maximization of the finite–SNR interference objective over all interferer power sequences $z_w[i]=|x_w[i]|^2$ satisfying
\[
\mathcal{F}_w=\bigl\{\underline{z_w}\in\mathbb{R}_+^n:\tfrac1n\sum_{i=1}^n z_w[i]=P_w\bigr\},\quad w\ne k,
\]
admits an optimizer of the form
\[
\underline{z_w}^\star = nP_w\,\mathbf e_{j_w^\star},\qquad j_w^\star\in\{1,2,\ldots,n\},
\]
i.e. each interferer concentrates its average power on a single symbol time.
Consequently, the evaluation of the finite–SNR upper bound in~\eqref{eq:optim} can be restricted to such one–sparse interferer user power allocations.
\end{prop}

\begin{proof}
Please refer to Appendix~\ref{proof:prop4}
\end{proof}

Thus, Proposition~\ref{prop:finiteSNR-extreme} implies that the outer bound from \eqref{eq:optim} can finally be rewritten as
\begin{align*}
nR_k &\leq n\log\biggl(\!1+\frac{P_k}{2\sigma_k^2} \biggl[1\!+\!\biggl(\sum_{\substack{w=1\\w \neq k}}^K P_w \sum_{m=-M}^M c_{k,w}^m \biggl)^2\biggl] \biggl)+n\epsilon_n,
\end{align*}
where the last step follows by Jensen's inequality, with $\epsilon_n \to 0$ as $n \to \infty$. Clearly, this upper bound is achievable by constant amplitude signalling for all the interferers $w \neq k$, along with a Gaussian codebook for user-$k$. Hence, we conclude that the capacity of user-$k$ is indeed characterized by \eqref{eq:capLBUBfinite}.

\section{Numerical Results} \label{sec:num}
The parameters used in our numerical experiments are listed in Table~\ref{table:param}, and we initially consider $K=3$ WDM channels. As noted earlier, empirical evaluation of the RP kernels for these parameters yields $c_{k,w}^m \geq 0$, so we focus on the finite-SNR characterization. Figure~\ref{fig:OBsnr} shows the AHO inner bound for user~2, which coincides with the single-user capacity of Proposition~\ref{thm:OBgenieKfinite}, as a function of the symmetric input power $P_1\!=\!P_2\!=\!P_3$. For comparison, we also plot the TIN achievable rates obtained when $X_k^n$ are i.i.d.\ circularly symmetric complex Gaussian random variables with variance $P_k$, together with the reference curve $\log(1+\mathrm{SNR})$ for $\mathrm{SNR}=P_2/(2\sigma_2^2)$. As previously noted, first-order RP models are not energy preserving and can exhibit superlinear pre-log behavior at high SNR \cite{keykhosravi2019accuracy}. The accelerated growth in the shaded region of Fig.~\ref{fig:OBsnr} reflects this artifact. These rates lie outside the validity range of the RP approximation (above $5.2$~dBm) (see Fig.~\ref{fig:OBregionSSFM}). Importantly, the advantage of AHO over TIN is already evident at moderate launch powers, well before the onset of this artifact, indicating that the observed gains are not solely a consequence of RP energy non-conservation.

\begin{figure}[t]
\vspace{-5.5cm}
\begin{center}
\scalebox{1.0}{\input{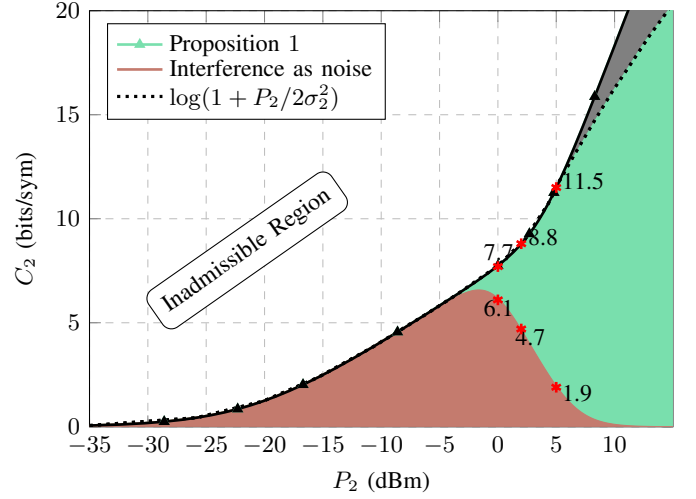}}
\end{center}
\caption{\footnotesize Comparison of the AHO inner bound of Prop.~\ref{thm:OBgenieKfinite} with the baseline TIN scheme as a function of input power. The red markers indicate the operating powers $0$, $2$, and $5$ dBm used to construct the three-dimensional achievable regions in Fig.~\ref{fig:OBregion}; the corresponding rates are annotated. The black dotted curve shows the $\log(1+\mathrm{SNR})$ benchmark for user~2 with $\mathrm{SNR}=P_2/(2\sigma_2^2)$ (see Sec.~\ref{sec:num}). The shaded gray region corresponds to power levels above $5.2$~dBm, where the RP model loses physical validity (cf. Fig.~\ref{fig:OBregionSSFM}). The plotted quantity is the single-user rate $C_k$; under symmetric time-sharing among $K$ users, the average per-user rate is $C_k/K$ (see Fig.~\ref{fig:succ_vs_tin_vs_aho} for $K=3$).}
\label{fig:OBsnr}
\vspace{-5mm}
\end{figure}

\begin{figure*}[!t]
\scalebox{0.9}{\input{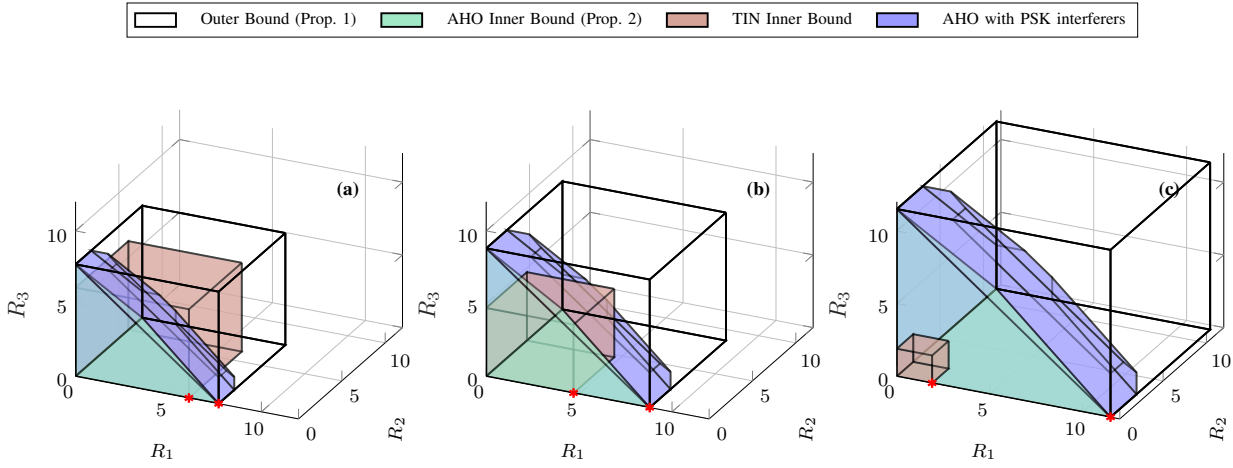}}
\caption{Achievable inner and outer bound regions for per-user transmit powers of (a) $0$ dBm, (b) $2$ dBm and (c) $5$ dBm. The intercepts of each region on the coordinate axes correspond to the red starred points shown in Fig.~\ref{fig:OBsnr}. The blue region is obtained using information-bearing 16-PSK interferers; unlike the triangular all-helps-one region, its pairwise rate-region projections are pentagonal, reflecting additional simultaneous-transmission operating points.}
\label{fig:OBregion}
\end{figure*}

\begin{figure}{}
\vspace{-5mm}
\begin{center}
\scalebox{0.64}{\input{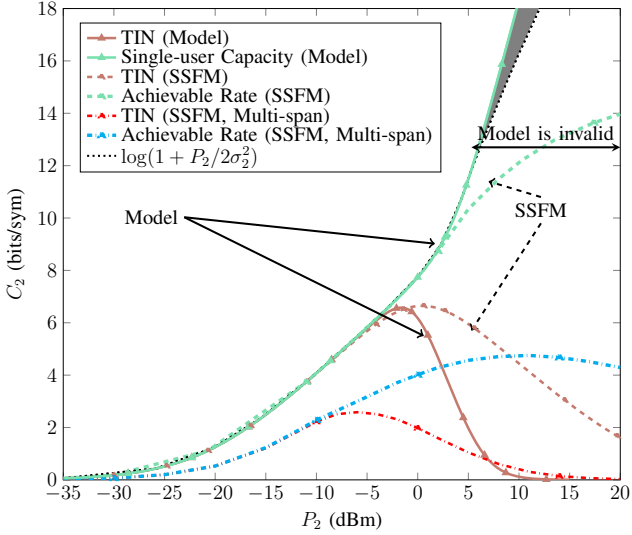}}
\end{center}
\caption{\footnotesize Capacity lower bounds based on the model \eqref{eq:Kuserapprox} (solid lines) and SSFM simulations (dashed lines). The shaded gray region corresponds to launch powers beyond approximately $5$\,dBm, where the RP model loses accuracy. The dash-dotted curves depict 3$\times$120\,km multi-span SSFM results; the CA-interferer advantage over TIN persists relative to single-span, confirming robustness to accumulated SPM and nonlinear distortion beyond RP validity.
}
\label{fig:OBregionSSFM}
\vspace{-4mm}
\end{figure}

Figures~\ref{fig:OBregion}(a)--(c) illustrate the rate trade-offs among the three users for equal input powers of $0$~dBm, $2$~dBm and $5$~dBm, respectively. The regions achievable under the \emph{treating interference as noise} (TIN) strategy, corresponding to Gaussian inputs for all users, are shown as the brown cubical regions. Formally, the achievable rate of user-$k$ under TIN can be written as the mutual information $R_k^{\text{TIN}} = I(X_k;Y_k)$, where the interference terms in~\eqref{eq:Kuserapprox} are modeled as additional Gaussian noise. Since the effective noise variance grows quadratically with input power, while the useful signal power grows only linearly, $R_k^{\text{TIN}} \to 0$ as $P_k \to \infty$. This explains the vanishing of the TIN regions in the strongly nonlinear regime.  

By contrast, employing constant-amplitude signaling for the interferers together with Gaussian codebooks for the user of interest (cf. Prop.~\ref{thm:OBgenieKfinite}) produces three corner rate triples in the three-dimensional rate space. Time-sharing among these corner points yields the AHO inner bound, depicted as the green polyhedral regions in Figs.~\ref{fig:OBregion}(a)--(c). It is evident that the AHO region strictly enlarges the set of achievable rate tuples compared to TIN, with the performance gap becoming more pronounced as the input power increases, as seen in the transition from Fig.~\ref{fig:OBregion}(a) to Fig.~\ref{fig:OBregion}(c). While the AHO region is derived using deterministic constant-amplitude interferers, it need not be restricted as such. The blue regions in Figs.~\ref{fig:OBregion}(a)--(c) are obtained using information-bearing 16-PSK interferers; unlike the triangular AHO region, its pairwise rate-region projections are pentagonal, reflecting additional operating points where multiple users simultaneously achieve nonzero rates.

While Theorem~\ref{thm:OBgenieK} is stated in the high-SNR regime $P_k\to\infty$, such asymptotic characterizations are standard in multi-user information theory~\cite{etkin2008gaussian}. Proposition~\ref{thm:OBgenieKfinite} extends the result to finite SNR under an empirically validated sign condition. Our results go beyond existing TIN inner bounds by identifying constant-amplitude interferer signaling as extremal and yielding a coordinate-wise matching AHO region, with the resulting characterization closely tracking achievable rates over practical launch powers 0--5\,dBm.

Fig.~\ref{fig:OBregionSSFM} compares the bounds of Proposition~\ref{thm:OBgenieK} and the TIN rates predicted by the RP model \eqref{eq:Kuserapprox} with split-step Fourier method (SSFM) simulations of the NLSE channel \eqref{eq:nlse}. Each user $k=1,2,3$ transmits an i.i.d.\ codeword $X_k^n$ with average power $P_k$. For the TIN benchmark, all users employ circularly symmetric complex Gaussian inputs $\mathcal{CN}(0,P_k)$, while for the lower-bound scenario, only the channel of interest ($k=2$) uses Gaussian signaling and the interfering users employ constant-amplitude $16$-PSK symbols. The propagated signals yield output sequences $Y_k^n$, from which the achievable rates $\frac{1}{n}\sum_{i=1}^n I(X_k[i];Y_k[i])$ are estimated. SSFM simulations use a uniform spatial step size of $0.01\,\mathrm{km}$.

Fig.~\ref{fig:OBregionSSFM} reiterates that Gaussian signaling for the desired user combined with constant-amplitude (CA) interferers yields higher achievable rates than when the interferers also employ amplitude-varying modulation. At $0$, $2$, and $5$ dBm, the achievable rate with CA interferers is approximately $6.6$, $7.7$, and $9.3$ bits/sym, respectively, compared with approximately $6.6$, $6.7$, and $6.1$ bits/sym under TIN. Thus, the advantage of CA interferer signaling is also visible in the full NLSE simulation.

To assess the robustness of this observation beyond the single-span configuration of Table~\ref{table:param}, Fig.~\ref{fig:OBregionSSFM} also includes SSFM results for a three-span $3 \times 120$ km link. Although the achievable rates decrease due to accumulated nonlinear distortion, CA 16-PSK interferers continue to outperform Gaussian interferers across the considered power range, with the relative gap between the two SSFM curves comparable to the single-span case. This confirms that the qualitative advantage of CA interferer signaling is not restricted to the specific single-span scenario used in the RP derivation.

The figure also shows that, at higher powers (beyond approximately $5.2$\,dBm), the SSFM rates fall below the RP predictions, reflecting the reduced accuracy of the simplified model \eqref{eq:Kuserapprox}. This threshold lies within the launch-power range typically associated with the onset of strong nonlinear distortion in dispersion-uncompensated WDM systems; prior studies \cite{dar2013properties,essiambre2010capacity} report near-optimal operation at only a few dBm per channel. Consequently, the agreement between the RP analysis and SSFM simulations until $5.2$\,dBm indicates that the qualitative conclusions persist in the underlying NLSE channel over a practically relevant operating regime.

\begin{remark}
The AHO scheme is coordinate-wise optimal, achieving the maximum rate for each user individually. However, these optima cannot be attained simultaneously, so time-sharing is required to obtain a multi-user trade-off. This highlights that the gap between inner and outer bounds arises from the simultaneity constraint, which is a structural consequence of the extremal input characterisation. Formally, the AHO region satisfies
$\sum_{k=1}^K R_k \leq \max_k C_k$ on its boundary, whereas the outer bound permits $R_k \leq C_k$ for each $k$.
\end{remark}

\begin{figure}{}
\vspace{-5mm}
\begin{center}
\scalebox{1.1}{\begin{tikzpicture}
\begin{axis}[
    width=0.92\linewidth,
    height=0.58\linewidth,
    xlabel={Number of users $K$},
    ylabel={Symmetric Rate (bits/sym)},
    xmin=1, xmax=20,
    ymin=0,
    grid=both,
    legend style={at={(0.02,0.98)},anchor=north west},
    tick align=outside,
]


\addplot[
    line width=1.5pt,
    orange,
    dashdotted,
    smooth,
    mark=none,
]
coordinates {
    (1,0.00) (2,1.40) (3,2.20) (4,2.70) (5,3.04)
    (6,3.33) (7,3.51) (8,3.67) (9,3.82) (10,3.92)
    (11,4.00) (12,4.08) (13,4.14) (14,4.20) (15,4.25)
    (16,4.30) (17,4.34) (18,4.38) (19,4.42) (20,4.46)
};
\addlegendentry{$C(K)\,(1-1/K)$}

\addplot[
    line width=1.5pt,
    green!60!black,
    dashdotted,
    smooth,
    mark=none,
]
coordinates {
    (1,2.00) (2,1.40) (3,1.10) (4,0.90) (5,0.76)
    (6,0.67) (7,0.59) (8,0.53) (9,0.48) (10,0.44)
    (11,0.40) (12,0.37) (13,0.34) (14,0.32) (15,0.30)
    (16,0.29) (17,0.27) (18,0.26) (19,0.24) (20,0.23)
};
\addlegendentry{$C(K)/K$}

\end{axis}
\end{tikzpicture}}
\end{center}
\caption{\footnotesize Scaling of the symmetric achievable rate and the gap to the outer bound as the number of WDM users $K$ increases.}
\label{fig:Kscale}
\vspace{-5mm}
\end{figure}
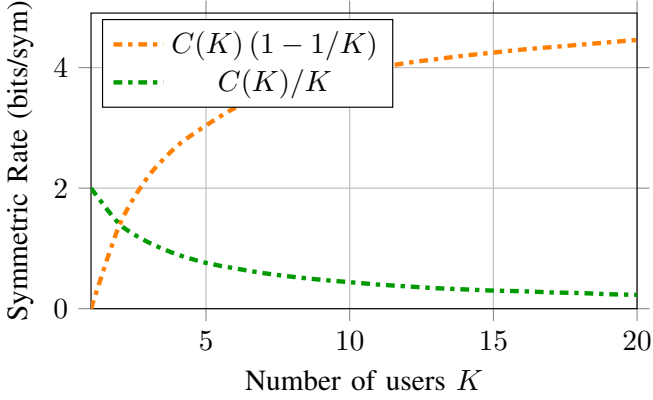

While the preceding results focus on $K=3$ for visualization purposes, the behavior for larger $K$ follows directly from the scaling of the single-user corner point. Let $C(K)$ denote the capacity corner point of Proposition~\ref{thm:OBgenieKfinite} for a symmetric $K$-user system. As shown in Fig.~\ref{fig:Kscale}, the symmetric rate achieved by the all-helps-one scheme scales as $C(K)/K$ and therefore decreases with $K$, whereas the gap to the Cartesian-product outer bound grows as $C(K)(1-1/K)$. Thus, increasing $K$ further accentuates the tension between single-user optimality and simultaneous multiuser operation.

We next consider a concrete multi-user decoding strategy. This allows us to benchmark the performance of simultaneous transmission against the time-shared all-helps-one scheme.

\textbf{Successive decoding--equalization scheme:}
For the sake of comparison, we also consider a successive decoding equalization strategy (SDE) for the RP interference channel with $K=3$ in \eqref{eq:Kuserapprox}. Fix the decoding order\footnote{Other orders are possible and lead to different rate tuples in general, but for the symmetric power setting here, the performance is order‑independent. In asymmetric scenarios, the optimal order decodes users with the highest effective signal‑to‑interference ratio first.} $3 \rightarrow 2 \rightarrow 1$. User~3 is decoded first while treating all nonlinear interference from users~1 and~2 as noise, yielding a rate $R_3 \le \frac{1}{n} I(X_3^n; Y_3^n)$ which coincides with treating interference as noise (TIN).

Upon successful decoding, the sequence $X_3^n$ becomes known at receivers~1 and~2. Conditioned on $X_3^n$, the effective channel for user~2 becomes
\[
Y_2[i]
= \bigl(1 + j S_{2}^{(3)}[i]\bigr) X_2[i]
+ j X_2[i] S_{2}^{(1)}[i]
+ N_2[i],
\]
where
\begin{align*}
S_{2}^{(3)}[i] := \sum_m c_{2,3}^m |x_3[i-m]|^2, \:
S_{2}^{(1)}[i] := \sum_m c_{2,1}^m |X_1[i-m]|^2 .
\end{align*}
The term $S_{2}^{(3)}[i]$ is known deterministically, while $S_{2}^{(1)}[i]$ remains unknown and acts as residual nonlinear interference. Equalization by the known factor $(1 + j S_{2}^{(3)}[i])$ yields
\[
\tilde Y_2[i]
= X_2[i]
+ j X_2[i] \tilde S_{2}^{(1)}[i]
+ \tilde N_2[i],
\]
which has the same structural form as the original channel but with only user~1 contributing interference. User~2 is then decoded at rate $R_2 \le \frac{1}{n} I(X_2^n; Y_2^n \mid X_3^n)$.

After decoding both users~3 and~2, receiver~1 knows $X_2^n$ and $X_3^n$, and hence all interference terms affecting user~1:
\[
S_1[i] = \sum_{w=2}^{3} \sum_m c_{1,w}^m |x_w[i-m]|^2 .
\]
The channel reduces to
\[
Y_1[i] = (1 + j S_1[i]) X_1[i] + N_1[i],
\]
which becomes an additive Gaussian channel with a known gain after equalization. The achievable rate is therefore $R_1 \le \frac{1}{n} I(X_1^n; Y_1^n \mid X_2^n, X_3^n).$

Figure~\ref{fig:succ_vs_tin_vs_aho} compares the sum rate of this scheme with TIN and the all-helps-one construction. The SDE strategy improves upon TIN by exploiting partial knowledge of decoded interferers to partially mitigate nonlinear distortion. At low launch powers, SDE also outperforms the all-helps-one scheme, since simultaneous transmission by all users efficiently utilizes the channel when nonlinear effects are weak. However, as power increases, residual stochastic nonlinear interference in the early decoding stages limits performance, causing SDE to fall below the all-helps-one curve. In contrast, the all-helps-one scheme avoids simultaneous interference entirely by operating at single-user optimal points via time sharing, which becomes increasingly advantageous in the strong-nonlinearity regime.

\begin{figure}
\vspace{-5mm}
\begin{center}
\scalebox{0.93}{\definecolor{mycolor1}{rgb}{0.76078,0.76078,0.76078}%
\definecolor{mycolor2}{rgb}{0.47412,0.87412,0.67863}%
\definecolor{mycolor3}{rgb}{0.92900,0.69400,0.12500}%
\definecolor{mycolor4}{rgb}{0.77000,0.48353,0.430000}%

\begin{tikzpicture}[tight background]
\begin{axis}[%
every axis/.append style={font=\small},
width=1.05\columnwidth,
height=0.8\columnwidth,
xmin=-35,
xmax=15,
xtick={-35,-30,-25,-20,-15,-10,-5,0,5,10},
xlabel={Launch Power (dBm)},
ymin=0,
ymax=10,
ylabel={Achievable Symmetric Rate (bits/sym)},
ylabel near ticks,
xlabel near ticks,
axis background/.style={fill=white},
xmajorgrids,
ymajorgrids,
legend style={
legend pos=north west,legend cell align=left,row sep=-0.5ex,grid style={dashed},font=\small
}
]

\addplot [color=mycolor2, mark=triangle*, mark size=1.5,  mark repeat=5, mark phase=4, line width=1.5pt]
  table[row sep=crcr, y expr=(1.1/3)*\thisrowno{1}]{%
-35	0.053910380168432\\
-32.4	0.0966460479040592\\
-30.3	0.153632400336519\\
-28.6	0.221818269368086\\
-27.1	0.304239080262668\\
-25.8	0.396918191503531\\
-24.6	0.503330111764171\\
-23.4	0.632592151892091\\
-22.3	0.773140291332751\\
-21.2	0.936101161926771\\
-20.1	1.12209180629173\\
-19	1.33100845377491\\
-17.9	1.56203932983848\\
-16.7	1.83760961301896\\
-15.4	2.16107127921973\\
-14	2.53439933963911\\
-12.5	2.95807411257401\\
-10.7	3.49165801940611\\
-8.6	4.139455761862\\
-6.1	4.93504978721244\\
-3.5	5.78595281140669\\
-1.90000000000001	6.33013478329129\\
-0.800000000000004	6.72629967049322\\
0	7.0367567691824\\
0.700000000000003	7.33403929250529\\
1.3	7.61685291758662\\
1.8	7.87909275655503\\
2.3	8.17161775038099\\
2.7	8.43124146498391\\
3.1	8.71630460520075\\
3.5	9.02852919862691\\
3.9	9.36866853711275\\
4.3	9.73639157278734\\
4.8	10.232672096258\\
5.3	10.765439340867\\
5.9	11.4450740803583\\
6.5	12.1592805172635\\
7.3	13.1500442646957\\
8.3	14.4279673604683\\
9.8	16.3872483510964\\
12.5	19.958935451836\\
15	23.2784399657178\\
};
\addlegendentry{All-helps-one scheme (AHO)}

\addplot [color=mycolor4, mark=triangle*, mark size=1.5,  mark repeat=5, mark phase=4, line width=1.0pt]
  table[row sep=crcr, y expr=(1.1/3)*\thisrowno{1}]{%
-35	0.053910380168297\\
-32.4	0.0966460479025812\\
-30.3	0.153632400326636\\
-28.6	0.221818269322618\\
-27.1	0.304239080089985\\
-25.8	0.396918190960939\\
-24.6	0.503330110220141\\
-23.4	0.632592147548934\\
-22.3	0.773140280242842\\
-21.2	0.936101133889998\\
-20.1	1.12209173605716\\
-19	1.33100827918214\\
-17.9	1.5620388982713\\
-16.7	1.83760845876827\\
-15.4	2.16106792815874\\
-14	2.5343887241094\\
-12.5	2.95803717330993\\
-10.7	3.49148945873452\\
-8.5	4.16971401712835\\
-4.2	5.50468229256688\\
-3.5	5.69446817286476\\
-2.90000000000001	5.83290944836092\\
-2.5	5.9056896839054\\
-2.1	5.95653882823546\\
-1.8	5.97639614871774\\
-1.5	5.97732843922363\\
-1.2	5.95645824132733\\
-0.899999999999999	5.91125349290134\\
-0.600000000000001	5.83983699775955\\
-0.300000000000004	5.74125742744292\\
0	5.61564883961075\\
0.299999999999997	5.46423095208318\\
0.600000000000001	5.28915156715762\\
1	5.0238206860684\\
1.4	4.72900117448969\\
1.9	4.33058680536951\\
2.6	3.74083316685753\\
3.9	2.63853165719195\\
4.5	2.16108379289557\\
5	1.79192098543657\\
5.4	1.52033521488747\\
5.8	1.27302585202779\\
6.2	1.05207874680286\\
6.6	0.858554732360687\\
7	0.692380908641127\\
7.4	0.55239845004575\\
7.8	0.436559332501297\\
8.2	0.342217723208407\\
8.7	0.250051907588009\\
9.2	0.181197750242958\\
9.8	0.122108981872884\\
10.5	0.0764417891746731\\
11.4	0.041526800219799\\
12.7	0.0170534630629575\\
14.8	0.00401478326309501\\
15	0.00349732023550331\\
};
\addlegendentry{Interference as noise (TIN)}

\addplot [color=red, mark=triangle*, mark size=1.5,  mark repeat=5, mark phase=4, line width=1.0pt]
  table[row sep=crcr, y expr=(1.1/3)*\thisrowno{1}]{%
-35	0.080\\
-32.4	0.145\\
-30.3	0.225\\
-28.6	0.320\\
-27.1	0.430\\
-25.8	0.560\\
-24.6	0.720\\
-23.4	0.900\\
-22.3	1.10\\
-21.2	1.33\\
-20.1	1.58\\
-19	1.85\\
-17.9	2.15\\
-16.7	2.50\\
-15.4	2.90\\
-14	3.35\\
-12.5	3.85\\
-10.7	4.50\\
-8.6	5.30\\
-6.1	6.20\\
-3.5	7.10\\
-1.9	7.70\\
-0.8	8.00\\
0	8.15\\
0.7	8.20\\
1.3	8.18\\
1.8	8.10\\
2.3	8.00\\
2.7	7.85\\
3.1	7.65\\
3.5	7.40\\
3.9	7.10\\
4.3	6.75\\
4.8	6.30\\
5.3	5.85\\
5.9	5.30\\
6.5	4.80\\
7.3	4.15\\
8.3	3.50\\
9.8	2.75\\
12.5	1.90\\
15	1.40\\
};
\addlegendentry{Successive decoding--equalization (SDE)}

\end{axis}
\end{tikzpicture}}
\caption{\footnotesize Comparison of achievable symmetric rates for $K=3$ users for AHO, TIN and SDE achievable schemes.}
\label{fig:succ_vs_tin_vs_aho}
\end{center}
\vspace{-4mm}
\end{figure}
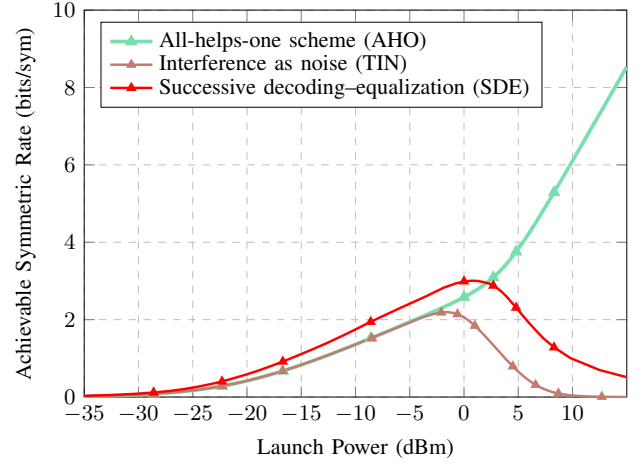

\section{Conclusions} \label{sec:conc}
We derive high-SNR outer and inner bounds for the multiuser WDM interference channel under \emph{average} power constraints within the first-order regular perturbative approximation. A genie–aided argument yields closed-form single-user capacities and identifies constant-amplitude interferers as extremal; under an empirically validated condition (and within the truncated RP model), a finite-SNR single-user capacity characterization is also obtained. The resulting time-shared \emph{all-helps-one} region strictly outperforms TIN. While our capacity characterizations are derived for the first-order RP model, the qualitative trends are validated by SSFM simulations of the NLSE channel over practical launch powers.

Future work includes quantifying the gap between the all-helps-one region and the true capacity region and incorporating additional effects such as SPM and signal-noise interaction. Extensions to dual-polarization and multi-span links are also of interest, with the broader goal of clarifying the limits of nonlinear optical networks and guiding system design.



\appendices
\section{Diagonal optimal input covariance} \label{app:diag}
\begin{lemma}[Diagonal optimal input covariance]\label{lem:diag-opt}
Fix a user $k$ and blocklength $n$, and consider the effective single-user channel given by
\[
Y_k^n = H_k X_k^n + N_k^n,
\]
where $H_k=\diag(h_{k,1},h_{k,2},\ldots,h_{k,n})$ and $N_k^n\sim\mathcal{CN}0,\Sigma_k)$ with $\Sigma_k=\diag(\sigma_{k,1}^2,\sigma_{k,2}^2,\ldots,\sigma_{k,n}^2)$. For Gaussian inputs $X_k^n\sim\mathcal{CN}(0,K_k)$ under the constraint $\tr(K_k)\le nP_k$, 
\[
I(X_k^n;Y_k^n)
=\log\det\!\big(I_n+\Sigma_k^{-1/2}H_k K_k H_k^{\!H}\Sigma_k^{-1/2}\big)
\]
is maximized by a diagonal covariance matrix $K_k^\star$.
\end{lemma}

\begin{proof}
Let $D_k := H_k^{\!H}\Sigma_k^{-1}H_k=\diag(d_{k,1},d_{k,2},\ldots,d_{k,n})$ with $d_{k,i}\ge 0$. Then we have
\[
I(X_k^n;Y_k^n)
=\log\det\!\big(I_n + D_k^{1/2} K_k D_k^{1/2}\big).
\]
Since $D_k^{1/2}K_kD_k^{1/2}$ is Hermitian positive semidefinite, Hadamard’s inequality implies
\[
\det\!\big(I_n + D_k^{1/2} K_k D_k^{1/2}\big)
\le \prod_{i=1}^n \big(1+d_{k,i}(K_k)_{ii}\big),
\]
with equality if and only if $D_k^{1/2} K_k D_k^{1/2}$ is diagonal. Hence, correlation across symbols (off-diagonal entries of $K_k$) cannot increase $I(X_k^n;Y_k^n)$. Replacing $K_k$ by its diagonal part preserves the trace constraint and does not decrease the mutual information, so an optimal diagonal covariance $K_k^\star$ exists.
\end{proof}

\section{Error bound for the high-SNR logarithmic approximation} \label{sec:appA}
\begin{lemma}[Uniform log-approximation bound]\label{lem:log_approx}
Fix user $k$ and blocklength $n$. Define the per-symbol powers
$P_{k,i}:=\mathbb{E}[|X_k[i]|^2]$ and
\[
S_k[i]:= \sum_{\substack{w=1\\ w\neq k}}^{K} |x_w[i-m]|^2 \sum_{m=-M}^{M} c_{k,w}^m.
\]
Let the per-symbol SNRs and its blockwise lower envelope be 
\begin{equation}\label{eq:gamma-def}
\gamma_{k,i}:= \frac{P_{k,i}}{2\sigma_k^2}\bigl(1+S_k[i]^2\bigr),
\qquad
\underline{\gamma}_k := \min_{1\le i\le n}\gamma_{k,i}.
\end{equation}
Then the following bound holds
\begin{equation}\label{eq:avg-log-gap}
0 \le \frac{1}{n}\sum_{i=1}^n \bigl[\log(1+\gamma_{k,i})-\log\gamma_{k,i}\bigr]
\le \frac{1}{\underline{\gamma}_k}.
\end{equation}
Consequently, along any sequence of power allocations $\{P_{k,i}\}$ such that $\underline{\gamma}_k\to\infty$,
\begin{equation}\label{eq:avg-log-replace}
\frac{1}{n}\sum_{i=1}^n \log(1+\gamma_{k,i}) = \frac{1}{n}\sum_{i=1}^n \log\gamma_{k,i} + o(1).
\end{equation}
\end{lemma}

\begin{proof}
For any $\gamma>0$,
\[
\log(1+\gamma)-\log\gamma =\log\!\Big(1+\frac{1}{\gamma}\Big) \le \frac{1}{\gamma}.
\]
Hence, for each $i$,
\[
0\le \log(1+\gamma_{k,i})-\log\gamma_{k,i} \le \frac{1}{\gamma_{k,i}} \le \frac{1}{\underline{\gamma}_k}.
\]
Averaging over $i=1,2,\ldots,n$ yields \eqref{eq:avg-log-gap}. If $\underline{\gamma}_k\to\infty$, then $1/\underline{\gamma}_k\to0$, which gives \eqref{eq:avg-log-replace}.
\end{proof}


\section{Proof of Proposition~\ref{prop:extreme}} \label{proof:prop3}
Note that each \(\mathcal{F}_w\) is the scaled probability simplex
\[
\mathcal{F}_w
= \mathrm{Conv}\bigl\{\,nP_w\,\mathbf{e}_j:\ j=1,2,\dots,n\,\bigr\},
\]
where \(\mathbf{e}_j\) is the \(j\)-th standard basis vector of \(\mathbb{R}^n\).  
The feasible set for the interferer power allocations is the product
\[
\mathcal{F}
:= \prod_{\substack{w=1\\w\ne k}}^{K}\mathcal{F}_w
\subset \mathbb{R}^{n(K-1)},
\]
which is a compact convex polytope (since it is given by a finite product of simplices).
The interference vector $\underline{z}$ (whose entries collect the terms $|x_w[i-m]|^2$ appearing in the effective channel for user~$k$) is an affine function of $(\underline{z_w})_{w\neq k}$. Since the weight matrix $C\succeq 0$ in the objective $Q(\underline{z})$ is positive semidefinite, $Q$ is convex in $\underline{z}$ and hence convex in $(\underline{z_w})_{w\neq k}$.

By the standard result that a convex function attains its maximum over a compact convex polytope at an extreme point~\cite[Thm.~18.5]{rockafellar1970convex}, any maximizer of $Q$ over $\mathcal{F}$ can be chosen at an extreme point of $\mathcal{F}$. Since each $\mathcal{F}_w$ is a simplex, its extreme points are $\{nP_w\mathbf{e}_j:\,j=1,2,\ldots,n\}$, and thus
\[
\mathrm{ext}(\mathcal{F}) =\bigl\{(\underline{z_w})_{w\neq k}:\; \underline{z_w}=nP_w\,\mathbf e_{j_w},\ j_w\in\{1,2,\ldots,n\}\bigr\}.
\]
Equivalently, each interferer $w\neq k$ concentrates its entire power $nP_w$ on a single coordinate $j_w$.

\section{Proof of Proposition~\ref{prop:finiteSNR-extreme}} \label{proof:prop4}
Let $\underline z=(\underline{z_w})_{w\ne k}$ with $z_w[i]=|x_w[i]|^2$.  
Define the linear map $L_k:\mathbb{R}^{n(K-1)}\to\mathbb{R}^n$ by
\[
(L_k\underline z)_i \;=\;
\sum_{w\ne k}\sum_{m=-M}^{M} c_{k,w}^m\, z_w[i-m],
\qquad i=1,2,\dots,n,
\]
using the convention $z_w[t]=0$ if $t\notin\{1,2,\dots,n\}$ and $c_{k,w}^m=0$ for $|m|>M$. The finite–SNR objective in~\eqref{eq:optim} can then be rewritten as (with $a_i=\frac{P_{k,i}}{2\sigma_k^2}$)
\[
F(\underline z) := \sum_{i=1}^n \phi_i\bigl((L_k\underline z)_i\bigr),
\:\:\:
\phi_i(y) := \log\bigl(1+a_i(1+y^2)\bigr),\;a_i>0.
\]

Each $\phi_i$ is differentiable and $L_k$ is linear, so $F$ has continuous first derivatives with respect to $\underline z$. By the chain rule, for any interferer $w\ne k$ and coordinate $j\in\{1,2,\ldots,n\}$,
\[
\frac{\partial F}{\partial z_w[j]}
= \sum_{i=1}^n
\phi_i'\!\bigl((L_k\underline z)_i\bigr)
\frac{\partial (L_k\underline z)_i}{\partial z_w[j]}.
\]
Expanding $(L_k\underline z)_i$,
\[
(L_k\underline z)_i
= \sum_{w'\ne k}\sum_{m=-M}^{M} c_{k,w'}^m\, z_{w'}[i-m].
\]
When differentiating with respect to $z_w[j]$, only the $w'=w$ terms remain:
\[
\frac{\partial (L_k\underline z)_i}{\partial z_w[j]}
= \sum_{m=-M}^{M} c_{k,w}^m\,
\frac{\partial}{\partial z_w[j]} z_w[i-m].
\]
Since $\partial z_w[t]/\partial z_w[j] = \mathbf 1\{t=j\}$,
\[
\frac{\partial (L_k\underline z)_i}{\partial z_w[j]}
= \sum_{m=-M}^{M} c_{k,w}^m\,\mathbf 1\{\,i-m=j\,\}
= c_{k,w}^{\,i-j},
\]
with the convention $c_{k,w}^r=0$ if $|r|>M$.

Define
\[
s[i] := \phi_i'\!\bigl((L_k\underline z)_i\bigr)
= \frac{2a_i (L_k\underline z)_i}{1+a_i(1+(L_k\underline z)_i^2)},
\]
so that
\[
{\displaystyle
\frac{\partial F}{\partial z_w[j]}
= \sum_{i=1}^n s[i]\; c_{k,w}^{\,i-j},\qquad j=1,2,\dots,n.}
\]

\vspace{-1mm}

Under the sign assumption $c_{k,w}^m\ge0$ and since $(L_k\underline z)_i\ge 0$ (powers are nonnegative), we have $s[i]\ge0$, so that each partial derivative $\partial F/\partial z_w[j]\ge0$.


Now fix an interferer index $w$. Since $F$ is differentiable and coordinate-wise non-decreasing (i.e., $\partial F/\partial z_w[j]\ge 0$ for all $j$), we apply the KKT conditions for maximizing $F$ subject to the linear sum and nonnegativity constraints. Introducing a Lagrange multiplier $\lambda_w\in\mathbb{R}$ for $\mathbf{1}^\top\underline{z_w}=nP_w$ and a nonnegative multiplier vector $\mu\in\mathbb{R}^n_+$ for $\underline{z_w}\ge 0$, the stationarity and complementary slackness conditions are
\begin{equation}\label{eq:KKT}
\nabla_{\underline{z_w}} F = \lambda_w \mathbf{1} - \mu, \qquad \mu \ge 0, \qquad \mu^\top \underline{z_w} = 0 .
\end{equation}

If a coordinate $j$ is active ($z_w[j]>0$), then $\mu_j=0$ and $\frac{\partial F}{\partial z_w[j]} = \lambda_w$. For inactive coordinates ($z_w[j]=0$), we have $\mu_j\ge 0$ and thus $\frac{\partial F}{\partial z_w[j]} \le \lambda_w$. Therefore, every positive entry of an optimizer must occur at indices where the partial derivative attains its maximum value $\lambda_w=\max_j \frac{\partial F}{\partial z_w[j]}$. Because the feasible set $\mathcal{F}_w$ is a simplex (convex hull of the points $nP_w\, \mathbf{e}_j$), any feasible $\underline{z}_w$ can be written as a convex combination of its vertices. Since $F$ is coordinate-wise non-decreasing and we have identified a coordinate $j_w^\star$ where $\frac{\partial F}{\partial z_w[j]}$
is maximal, allocating the entire power $nP_w$ on that coordinate is in fact optimal. Therefore, we conclude that there exists an optimizer of the form
\begin{equation*}
\underline{z_w}^\star = nP_w\,\mathbf e_{j_w^\star}, \qquad
j_w^\star \in \arg\max_j \frac{\partial F}{\partial z_w[j]} .
\end{equation*}
Applying the same argument independently to each interferer $w\neq k$ yields a joint maximizer in which every $\underline{z_w}$ is one-sparse.\footnote{The joint decoupling follows because $\mathcal{F} = \prod_{w \neq k} \mathcal{F}_w$ is a Cartesian product of independent simplices: at any joint maximiser, for each $w$ there exists $\lambda_w$ satisfying the per-user KKT conditions; the coupling through $(L_k \underline{z})_i$ affects only the value of $\lambda_w$, not the one-sparse conclusion.} The constant‑amplitude assignment $\underline{z}_w = P_w \mathbf{1}$ yields the same value of $F(\underline{z})$ as the one‑sparse maximizer because $F$ depends on $(L_k \underline{z})_i$ only through the sum of powers across symbols. Hence, the upper bound is achievable.

\nocite{}
\bibliographystyle{IEEEtran}
\bibliography{mylitoptics.bib}

\end{document}